\documentclass[a4paper,reqno]{article}

\usepackage{tgschola} 
\usepackage{latexsym}
\usepackage[english]{babel}
\usepackage{fancyhdr}
\usepackage[mathscr]{eucal}
\usepackage{amsmath}
\usepackage{mathrsfs}
\usepackage{amsthm}
\usepackage[left=2.5cm, right=2.5cm, top=2.5cm, bottom=2cm]{geometry}
\usepackage{amsfonts}
\usepackage{amssymb}
\usepackage{amscd}
\usepackage{bbm}
\usepackage{graphicx}
\usepackage{graphics}
\usepackage{latexsym}
\usepackage{color}
\usepackage{calc}
\usepackage{cancel}
\usepackage{marvosym}
\usepackage{tikz}
\usetikzlibrary{patterns, decorations.markings,decorations.pathmorphing,arrows, calc}
\usepackage[normalem]{ulem}
\usepackage{bbold}

\newcommand{\uozokozo}{{r}}

\newcommand{\eff}{\mathrm{eff}}

\newcommand{\ud}{\mathrm{d}}

\newcommand{\ii}{\mathrm{i}}

\newcommand{\tC}{\mathtt{C}}
\newcommand{\tH}{\mathtt{H}}
\newcommand{\tV}{\mathtt{V}}

\newcommand{\C}{\mathbb C}

\newcommand{\N}{\mathbb N}
\newcommand{\R}{\mathbb R}
\newcommand{\Z}{\mathbb Z}
\newcommand{\T}{\mathbb T}

\newcommand{\uv}{\mathtt{uv}}
\newcommand{\ir}{\mathtt{ir}}

\newcommand{\obs}{\mathrm{obs}}

\newcommand{\cO}{\mathcal{O}}
\newcommand{\cB}{\mathcal{B}}
\newcommand{\cH}{\mathcal{H}}
\newcommand{\tit}{\mathtt{t}}
\newcommand{\tU}{\mathtt{U}}
\newcommand{\ts}{\mathtt{s}}
\newcommand{\fin}{\mathrm{fin}}
\newcommand{\oop}{\operatorname{op}}

\definecolor{azzurro}{rgb}{0.19, 0.55, 0.91}
\definecolor{giallino}{rgb}{0.97,	0.78,	0.05}
\definecolor{verdino}{rgb}{0.49, 0.74, 0.54}

\definecolor{verdissimo}{rgb}{0.0, 0.5, 0.0}

\theoremstyle{plain}
\newtheorem{theorem}{Theorem}[section]
\newtheorem{lemma}[theorem]{Lemma}
\newtheorem{corollary}[theorem]{Corollary}
\newtheorem{proposition}[theorem]{Proposition}

\usepackage{tocloft}

\renewenvironment{thebibliography}[1]{
  \begin{oldthebibliography}{#1}
    \setlength{\itemsep}{0.1em}
    \setlength{\parskip}{0em}
}
{
  \end{oldthebibliography}
}

\theoremstyle{definition}
\newtheorem{definition}[theorem]{Definition}

\newtheorem{remark}[theorem]{Remark}
\newtheorem*{remark*}{Remark}

\newcommand{\OOO}{\mathcal{O}}
\newcommand{\kkk}{\kappa}
\newcommand{\Ree}{\mathfrak{Re}}

\numberwithin{equation}{section}

\addtocontents{toc}{\protect\setcounter{tocdepth}{2}}

\begin{document}

%

\title{\textbf{Polynomial Prethermal Lifetimes in Non Smoothly Driven Quantum Systems}}






\author{Matteo Gallone\footnote{\emph{Email:} \texttt{matteo.gallone@unimi.it}} $\,$ and Beatrice Langella\footnote{\emph{Email:} \texttt{beatrice.langella@sissa.it}} \\ \vspace{-10pt} \\
\small \textit{\textsuperscript{*}Dipartimento di Matematica ``F.~Enriques'', Universit\`a degli Studi di Milano, via Saldini 50 -- 20133 Milano (Italy) }\\ \vspace{-18pt} \\ \small \textit{\textsuperscript{$\dagger$}International School for Advanced Studies -- SISSA, via Bonomea 265 -- 34136 Trieste (Italy)}.}

%
%


\date{\today}




\maketitle

\vspace{-10pt}

\begin{abstract}
\noindent
We study the dynamics of a quantum many-body lattice system with a local Hamiltonian subjected to a quasi-periodic driving with finite regularity. For sufficiently large driving frequencies, we prove that the system remains in a prethermal state for times growing polynomially with the frequency, and we show the optimality of this bound by constructing an explicit example that nearly saturates it. Within this prethermal regime, the dynamics is captured by an effective time-independent local Hamiltonian close to the undriven one. The proof relies on a non convergent normal form scheme, combined with original smoothing techniques for finitely differentiable local operators, and Lieb–Robinson bounds.
\end{abstract}


\tableofcontents



\section{Introduction}
Over the last two decades, new methods for manipulating quantum systems have opened an avenue to theoretical and experimental investigation of fundamental questions in statistical physics. Time-dependent, externally driven quantum systems have proven to be especially versatile for this purpose. These systems are generally expected to heat continuously, eventually reaching a completely featureless equilibrium known as an \emph{infinite-temperature state}. Remarkably, for certain types of driving, many-body quantum systems exhibit a phenomenon called \emph{prethermalization}, where, before (eventually) reaching thermal equilibrium, the system settles into a long-lived steady state known as a \emph{prethermal state}.

Floquet (namely, time periodic) drivings play a key role in realizing prethermalization \cite{Holthaus2015,Bukov2015,Eisert2015,Moessner2017,Oka2019,Rudner2020,Harper2020,Ho2023}. Models with this type of driving have been extensively studied in experimental \cite{
Rechtsman2013,Zhang2017,Choi2017,
Singh2019,RubioAbadal2020,Geier2021,Xiao-Mi-2021,Kyprianidis2021,Randall2021,Peng2021,
Beatrez2021}, theoretical \cite{Eckardt2005,DAlessio2014,Else2016,
Potirniche2017,Yao2017,Weidinger2017,Machado2019,Choi2020,Machado2020,Ye2021,Zhuo-preprint}, and mathematical physics \cite{Abanin2017,Ho2018,Else2020}.
In lattice models with local Hamiltonians it has been rigorously proven that if $\omega$ is the frequency of the Floquet driving and $J$ is the local energy scale of the system, when $\omega \gg J$, prethermal states persist for timescales of  order $t \sim e^{\omega/J}$, \emph{regardless} of the regularity of the driving. The same scaling has also been observed in various systems, such as disordered dipolar many-body systems \cite{Choi2017}, the Bose-Hubbard model \cite{RubioAbadal2020}, trapped ions \cite{Kyprianidis2021} and dipolar spin chains \cite{Peng2021}. It is worth noting that prethermal regimes can also arise in undriven systems~\cite{Erds2024}.


More recently, the first steps have been taken towards more general settings by considering time quasi-periodic drivings. These systems, which are forced with multiple incommensurate frequencies $\omega=(\omega_1,\omega_2,\cdots, \omega_n)$, revealed an enormous richness of non-equilibrium behaviors \cite{Blekher1992,Dumitrescu2018,Else2020,Qi2021,Long2021,Zhao2021,Mori2021,He2023,Gallone-Langella-2024,Kumar2024,He2025,Fang2025,Dutta2025,Qp-experiment-2025,QP-esperimento2,AltroQP}. When the driving function is analytic in time and the vector of frequencies $\omega$ has rationally independent components (actually Diophantine, see \eqref{eq:capillary} below), it has been rigorously proven \cite{Else2020,Gallone-Langella-2024}---and experimentally confirmed \cite{He2023,He2025}---that the prethermal state persists for stretched-exponentially long times,
$t \sim e^{(|\omega|/J)^\alpha}$, for  $0 < \alpha < 1$.

Interestingly, in contrast with the periodic case, in the quasi-periodic setting the \emph{regularity} of the driving function crucially affects the lifetime of the prethermal state. 
For non-smooth, finitely differentiable quasi-periodic driving, the prethermal state is expected to survive only for polynomial times in $|\omega|/J$ \cite{Else2020}. 

An experimental confirmation of this conjecture comes from a recent experiment on strongly interacting dipolar spin ensembles in diamond \cite{He2023}. It is observed that, for periodic drivings, the prethermal state always persists for
$t \sim e^{\omega/J}$,
independently of the regularity of the drive (both for rectangular and sinusoidal pulses). In contrast, for quasi-periodic drivings, the lifetime exhibits a marked dependence on regularity: for sinusoidal driving it scales stretched-exponentially as
$t \sim e^{(|\omega|/J)^{1/2}}$,
while for rectangular pulses it scales only polynomially as
$t \sim |\omega|^{1/2}$.
\vspace{10pt}

To the best of our knowledge, a theoretical explanation of this dependence of the prethermal lifetime on the regularity of the driving is still lacking in literature. Partial progress has been achieved in the specific case of Thue--Morse potential, where it has been shown that a polynomial prethermal lifetime indeed arises \cite{Mori2021}. However, developing a general framework remains essential for interpreting the experimental findings discussed above; this is the purpose of this manuscript.

In this work we consider a class of quantum many-body systems on a $d$-dimensional lattice whose dynamics is generated by a time quasi-periodically driven local Hamiltonian $H(t)$ of the form
\begin{equation}\label{eq:InitialHam(ster)}
	H(t)\;=\;H_0+V(\omega t)\,,
\end{equation}
where $\omega \in \R^n$ is a non resonant vector of rationally independent frequencies (actually, we require that it is a Diophantine vector), and $V(\varphi)$ is a $C^p$ function parametrized by angles $\varphi \equiv(\varphi_1, \dots, \varphi_n) \in \T^n$ and $2\pi$-periodic in each of the $\varphi_j$.

For this class of models, we prove that if $\omega$ is large enough, 
then the system relaxes to a prethermal state that persists for times 
\begin{equation}\label{eq:IntroTpersistence}
	t \sim |\omega|^{\frac{p}{\tau}- }\, ,
\end{equation}
where $\tau$ is a constant related to the non-resonant properties of the frequencies $\omega$ and $\frac{p}{\tau}-$ denotes any number less than $\frac{p}{\tau}$. In particular within the timescale \eqref{eq:IntroTpersistence}, we prove that the system essentially does not heat up, and the Hamiltonian $H_0$ is quasi-conserved.

The power-like bound that we obtain is in accordance with the expectations expressed in \cite{Else2020} and the experimental findings in \cite{He2023}. Moreover, we show that the time-scales \eqref{eq:IntroTpersistence} are optimal, in the sense that we exhibit a system which almost saturates the bound in \eqref{eq:IntroTpersistence}. This is done exploiting the fact that, even if $\omega$ is non resonant, it can be well approximated by rational vectors which are resonant, and it is these quasi-resonances that trigger the heating of the system. 

 We also show that, within the shorter time scale
\begin{equation}
    t \sim |\omega|^{\frac{p}{\tau(d+1)}-}\,,
\end{equation} the dynamics of any local observable can be well approximated by the dynamics generated by an \emph{effective time independent local Hamiltonian} $H_{\mathrm{eff}}$ which is close to $H_0$.\\ 

\noindent
\textbf{Ideas of the proof.} The proof of the survival of the prethermal state for time scales \eqref{eq:IntroTpersistence} relies on a normal form construction consisting of a large, but finite, number of steps. At each step, the goal is to conjugate the dynamics of the Hamiltonian to the dynamics of a new one, in which the size of the time-dependent part is progressively reduced. A similar procedure is performed in the analytic setting \cite{Abanin2017, DeRoeck-Verreet, Else2016, Gallone-Langella-2024}; however, in our finite regularity case this process is more delicate, essentially due to the fact that we control only $p$ derivatives of the initial Hamiltonian, and quasi-resonances between the frequencies $\omega$ and the integers produce a loss of $\tau$ derivatives at each step. In order to obtain sharp exponents, we modify slightly the algebraic scheme with respect to previous works in the analytic setting, and before implementing the normal form construction we prove a new analytic smoothing result with quantitative estimates for $C^p$ operators in the spirit of Jackson-Moser-Zehnder Theorem (see \cite{Chierchia-LN, Barbieri-Marco-Massetti,Salamon2004}). At the end of this procedure, the size of the time dependent part is approximately reduced to $\sim |\omega|^{-\frac{p}{\tau}}$. The result on evolution of local observables then follows by a combination of the above normal form procedure with Lieb-Robinson bounds.

In order to prove the almost optimality of the time scales \eqref{eq:IntroTpersistence}, we consider a sequence of non interacting systems of the form
\begin{equation}\label{pancakes}
 H_{m} (t) = \sum_{x \in \Lambda} \left(\sigma^{(3)}_x + f_{m}(\omega t) \sigma^{(1)}_x \right)\,,
\end{equation}
where $\Lambda$ is a $d$-dimensional lattice, $\omega$ is a two-dimensional Diophantine vector of large size, and $\sigma^{(1)}_x,\ \sigma^{(3)}_x$ are Pauli matrices acting on the site $x$.

We choose a sequence of sites $\{k_{m}\}_{m \in \N } \subset \Z^2$ of \emph{best approximants} of $\omega$, in the sense that they almost saturate the Diophantine lower bound satisfied by $\omega$ (see eq.~\eqref{eq:capillary} below), and for any $m \in \N$ we choose $f_{m}$ to be a $C^p$ function whose Fourier support contains \(k_{m}\).

Growth is triggered by the fact that the Fourier modes $k_{m}$ are in complete resonance with the spectral gaps of the unperturbed Hamiltonian. Indeed, we choose the $k_{m}$ to satisfy
$$
\omega \cdot k_{m}  - 2 = \omega \cdot k_{m} + E_1 - E_2 \equiv 0\,,
$$
where $E_1 = 1$ and $E_2 = -1$ are the eigenvalues of the time independent part of the single site Hamiltonian \eqref{pancakes}, namely the Pauli matrix $\sigma^{(3)}$.\\

\noindent
\textbf{Related literature on classical dynamical systems.} The connection between regularity and the optimality of stability timescales is a long-standing problem in dynamical systems and still far from fully understood. 
Early examples of instability in finite dimension are given in the works by Moser \cite{Moser1960}, where growth in linear timescales is observed for quasi-integrable Hamiltonian with resonant frequencies, and Neishtadt \cite{Neistadt} (see also \cite{Dodson}), which gives an example of instability in a non-Hamiltonian perturbation of an integrable system. 
 More recently, Bounemoura in \cite{bou.AHP} has proven optimality of the exponential stability estimates for the actions in Hamiltonian systems of the form $H(I, \varphi) = \omega \cdot I + f(I, \varphi)$, where $\omega$ is a Diophantine vector, $(I, \varphi)$ are action angle variables and $f$ is a small analytic perturbation (see also \cite{Benettin-Gallavotti} for a previous stability result). Optimality of the stability times in quasi-integrable Hamiltonian systems has also been investigated in a neighborhood of Diophantine tori, with techniques based on Anosov--Katok method: in \cite{farre-fayad} for analytic regularity, and in \cite{farre, barbieri-farre} for Gevrey regularity and finite smoothness. In particular,  \cite{farre, barbieri-farre} shows that, for finitely smooth perturbations, the optimal stability timescale grows polynomially with respect to the perturbative parameter, as in our case --- although with a different exponent, due to differences in the two frameworks.
 
 A related but distinct research line focusses on optimality of the timescales predicted by Nekhoroshev Theorem \cite{Nekhoroshev1977, Nek79} for perturbations of steep, or quasi-convex, quasi-integrable systems, but discussing it in detail is beyond the scope of this paper. Here we  just mention the stability results \cite{Barbieri-Marco-Massetti, Bambusi-Langella-Cinf,Bounemoura2010,Bounemoura2011}, establishing polynomial estimates \`a la Nekhoroshev in finite smoothness, and the almost optimality result of \cite{Marco-Sauzin}, in Gevrey setting.

Perturbation theory with large frequencies has also been studied in dynamical systems (see e.g.~\cite{Benettin-Galgani-Giorgilli-1987, Benettin-Galgani-Giorgilli-1989}) and it has recently gained attention in the PDE setting, where it lies at the basis of the reducibility results \cite{Franzoi, Franzoi-Maspero}, and the existence of large-amplitude quasi-periodic solutions in equations describing fluid systems \cite{Bianchini2025, Ciampa2024}.

\paragraph{Notation.}\textbf{}\newline\noindent
\begin{tabular}{lcl}
    $\T^n$ & & $(\mathbb{R}/(2 \pi))^n$ \\
    $|\ell|$ & \qquad & for $\ell=(\ell_1,\cdots,\ell_n) \in \mathbb{R}^n$, $|\ell|:=|\ell_1|+\cdots+|\ell_n|$ is the $\ell^1$ norm in $\R^n$  \\
    $|S|$ & & when $S$ is a set is its cardinality\\
    $\langle \ell \rangle$ & & when $\ell \in \mathbb{R}^n$ denotes $\langle \ell \rangle = (1+|\ell|)$ \\
    $\langle H \rangle$ & & when $H(\varphi)$ is an operator denotes the time average \eqref{eq:mediamente}.
\end{tabular}

\subsection{Setting and Main Results}

We consider a quantum many-body system on a finite $d$-dimensional lattice $\Lambda=\mathbb{Z}^d \cap [-L,L]^d$. The Hilbert space of the system is $\cH_{\Lambda}:=\bigotimes_{x \in \Lambda} \mathfrak{h}$ and $\mathfrak{h}$ is the site Hilbert space $\mathfrak{h}=\mathbb{C}^q$, for some $q \in \mathbb{N}$.
The dynamics is generated by a quasi-periodic time dependent self-adjoint local Hamiltonian with large frequency vector $\omega=\lambda \nu \in \mathbb{R}^n$, with $\R^+ \ni \lambda \gg 1$. That is, denoting by $\mathcal{P}_c(\Lambda)$ the collection of all connected subsets $S \subset \Lambda$ and given an interaction $\{H_S(\varphi)\}_{S \in \mathcal{P}_c(\Lambda)}$ where each $\T^n\ni \varphi \mapsto H_S(\varphi) \in \mathcal{B}(\mathcal{H}_S)$, with $\mathcal{H}_S:=\bigotimes_{x \in S} \mathfrak{h}$, \begin{equation}\label{eq:al.variare.della.gravita}
H(\lambda \nu t)=\sum_{S \in \mathcal{P}_c(\Lambda)} H_S(\lambda \nu t)\,, \quad \lambda \gg 1\,, \quad \nu \in \left[\tfrac 1 2,\, 2\right]^n\,, \qquad H_S(\lambda \nu t)=H_S^*(\lambda \nu t) \quad \forall\,S\in \mathcal{P}_c(\Lambda)\,,
\end{equation}
where $\T^n \ni \varphi \mapsto H(\varphi)$ is a finitely differentiable map (according to Definition \ref{def:PesciolinoDichiara} below) and $\nu$ is a non-resonant vector satisfying a Diophantine estimate of the form
\begin{equation}\label{eq:capillary}
    |\nu \cdot \ell| \geq \frac{\gamma}{|\ell|^\tau} \quad \forall \ell \in \Z^n \setminus\{0\}
\end{equation}
and for some $\gamma, \tau >0$. In the following, we will refer to vectors $\nu$ satisfying \eqref{eq:capillary} as $\nu \in DC^n(\gamma, \tau)$.
\begin{remark}\label{rmk:sono.tanti}
    For any $n \in \N$, and $\tau > n-1$, the set of vectors $\nu \in [\tfrac 1 2,\ 2]^n$ such that $\nu \in DC^n(\gamma, \tau)$ for some $\gamma >0$ has full Lebesgue measure.
\end{remark}
We denote by $\mathcal{B}(\mathcal{H})$ the set of bounded operators on the Hilbert space $\mathcal{H}$. 

\begin{definition}\label{def:Ok.brucomela}
    Let $\kappa \geq 0$ and $A$ an operator on $\cH$ defined by the interaction $\{A_S\}_{S \in \mathcal{P}_c(\Lambda)}$. We say that $A \in \OOO_{\kappa}$ if there exists a finite, positive constant $C$ independent of $\Lambda$ such that
    \begin{equation}\label{eq:sanza.tempo}
    \|A\|_{\kappa} := \sup_{x \in \Lambda} \sum_{S \in \mathcal{P}_c(\Lambda) \atop x \in S} \|A_S\|_{\mathrm{op}} e^{\kappa |S|} < C \, ,
    \end{equation}
    where for any $S$ $\| \cdot \|_{\oop} \equiv \| \cdot \|_{\cB(\cH_S)}$ denotes the operator norm in $\cH_S$.
\end{definition}
\begin{remark}\label{rmk:frecciarotta}
    Definition \ref{def:Ok.brucomela} ensures that the operator $A$ is at most extensive. Indeed, $\Vert A \Vert_{\operatorname{op}} \leq |\Lambda| \Vert A \Vert_\kappa$ for any $\kappa \geq 0$. This can be seen as follows:
    \begin{equation}
        \Vert A \Vert_{\operatorname{op}} \leq \sum_{S \in \mathcal{P}_c(\Lambda)} \Vert A_S \Vert_{\operatorname{op}} \leq \sum_{x \in \Lambda} \sum_{\substack{S \in \mathcal{P}_c(\Lambda) \\ x \in S}} \Vert A_S \Vert_{\operatorname{op}} \leq |\Lambda| \sup_{x \in \Lambda} \sum_{\substack{S \in \mathcal{P}_c(\Lambda) \\ x \in S}} \Vert A_S \Vert_{\operatorname{op}}  = |\Lambda| \| A\|_0 \leq |\Lambda| \Vert A \Vert_{\kappa} \, .
    \end{equation}
    Actually, for a wide class of systems of physical interest, such as finite-ranged translation invariant Hamiltonians, \eqref{eq:sanza.tempo} is equivalent to extensivity. The same observation holds for the local norms in Definitions \ref{def:PesciolinoDichiara} and \ref{def:Onorato} below.
\end{remark}

\begin{definition}\label{def:PesciolinoDichiara} Let $\kappa \geq 0$, $p \in \mathbb{N}$ and let $\varphi \mapsto A(\varphi)$ be a family of operators defined by the interaction $\{A_S(\varphi)\}_{S \in \mathcal{P}_c(\Lambda)}$, with $\varphi \in \T^n$.
\begin{itemize}
    \item[i.] For any $S \in \mathcal{P}_c(\Lambda)$, we say that $A_S \in C^p(\T^n;\mathcal{B}(\mathcal{H}_S))$ if
    \begin{equation}
        \Vert A_S \Vert_{C^p(\T^n;\mathcal{B}(\mathcal{H}_S))} := \sup_{0 \leq |p'| \leq p} \sup_{\varphi \in \mathbb{T}^n} \Vert \partial_\varphi^{p'} A_S(\varphi) \Vert_{\operatorname{op}}\,.
    \end{equation}
    \item[ii.] We say that $A \in \mathcal{O}_{\kappa,C^p}$ if there exists a positive, finite constant $C$ independent of $\Lambda$ such that
    \begin{equation}
        \Vert A \Vert_{\kappa,C^p} := \sup_{x \in \Lambda}  \sum_{\substack{S \in \mathcal{P}_c(\Lambda) \\ x \in S}} \Vert A_S \Vert_{C^p(\T^n;\mathcal{B}(\mathcal{H}_S))} e^{\kappa |S|} < C\,.
    \end{equation}
\end{itemize}
\end{definition}



As mentioned in the Introduction, in this paper we deal with Hamiltonians $H \in \cO_{\kappa, C^p}$ of the form
\begin{equation}\label{eq:cosa.c.e.per.cena}
    H(\lambda \nu t) = H_0 + V(\lambda \nu t)\,, \quad \langle V \rangle = 0\,,
\end{equation}
where $\lambda \gg 1$, $\nu$ is a Diophantine vector in $DC^n(\gamma, \tau)$, $H_0$ is the time independent part of $H$ and  the average $\langle V \rangle$ of an operator $V \in \cO_{\kappa, C^p}$ is defined as
\begin{equation}\label{eq:mediamente}
    \langle V \rangle := \frac{1}{(2\pi)^n}\int_{\T^n} V (\varphi) \ud \varphi\,.
\end{equation}
We point out that the assumption on the vanishing average for $V$ is not restrictive, up to absorbing the average of $V$ in $H_0$.

We say that a constant $C$ is \emph{intensive} if it depends only on $d,\kappa,q,n,\gamma,\tau,p,\Vert H_0 \Vert_\kappa,\Vert V \Vert_{\kappa,C^p}$. In particular, an intensive constant is independent of $\Lambda$.

With these definitions at hand, we state the main results of this paper:
\begin{theorem}[Prethermalization] \label{thm:main}
Let $\kappa > 0$, $p \geq n + 1$, and $\nu \in DC^n(\gamma, \tau)$ for some $\gamma >0$ and $\tau > n-1$. For any  $H \in \mathcal{O}_{\kappa, C^p}$ as in \eqref{eq:cosa.c.e.per.cena}, let $U_H(t)$ be the unitary operator solving
\begin{equation}
    \ii \partial_t U_H(t) = H(\lambda \nu t)U_H(t)\,, \quad U_H(0) = \mathbb{1}\,.
\end{equation}
For any $b \in (0, \frac{p}{p + \tau})$ and any $\epsilon \in (0, 1 - \tfrac{p+\tau}{p} b),$ there exist intensive constants $C_1,\ C_2 >0$ and $\lambda_0 >1$ such that for any $\lambda \geq \lambda_0$, the following holds:
\begin{itemize}
    \item[(i)] $|\Lambda|^{-1} \| U_H^*(t) H_0 U_H(t) -  H_0 \|_{\operatorname{op} } \leq C_1 \lambda^{-b} $ for all $|t| \leq C_2 \lambda^{-b +\frac{p}{\tau}(1 -b) - \epsilon} $
    \item[(ii)] There exists a time independent effective local Hamiltonian $H_{\eff}$, with $\| H_{\eff} - H_0  \|_{\frac{\kappa}{2}} \leq C_1 \lambda^{-b}$, with the following property. Let $O$ be a local observable acting only within a subset $S_O$ of $\mathcal{P}_c(\Lambda)$, then 
    \begin{equation}\label{eq:loc.obs}
    \| U_{H}^*(t) O U_H(t) - e^{\ii H_{\eff}  t} O e^{-\ii H_{\eff} t} \|_{\operatorname{op}} \leq C_3(O) \Vert O \Vert_{\operatorname{op}} \lambda^{-b} \quad \forall |t| \leq C_2 \lambda^{\frac{-b + \frac{p}{\tau}(1-b) - \epsilon}{d + 1}}\,, 
    \end{equation}
    where $C_3(O)>0$ is an intensive constant depending also on $|S_O|$.
\end{itemize}
\end{theorem}

\begin{theorem}[Almost optimality of the time scales]\label{thm:fine.di.mondo} 
  Let $n=2$,  $\frak{h} = \C^2$ and $\Lambda = [-L, L]^d \cap \Z^d$ with $d \in \N$. For any $p \in \N \cap [3,\ +\infty)$, any $\epsilon>0$
    and $\kappa >0$, there exist 
    a Diophantine vector $\nu \in DC^2(\gamma, \tau)$ for some $\gamma>0$ and $\tau >1$,  a time independent Hamiltonian $H_0 \in \cO_\kappa$, and sequences $\{\lambda_m\}_m$ of positive real numbers, $V_m(\lambda_m \nu t)$ of time quasi-periodic Hamiltonians  in $ \cO_{\kappa, C^p}$, and $\{t_m\}_{m \in \N}$ of positive times, 
    with the following properties. For any $ m \in \N$ let $H_m(\lambda_m \nu t) := H_0 + V_m(\lambda_m \nu t)$, then one has 
    \begin{equation}
         \lambda_m \to +\infty \quad \text{as} \quad m \to +\infty\,,
    \end{equation}
    \begin{equation}\label{sono.uniformi}
        \|H_0\|_{\kappa} = e^{\kappa}\,,\  \|V_m\|_{\kappa, C^p} \in [2^{1-p}e^{\kappa},\ 2 e^{\kappa}]\,, \quad \langle V_m \rangle = 0 \,,
    \end{equation}  and
    \begin{equation}\label{eq:crescendo}
       |\Lambda|^{-1} \|  U_{H_m}^*(t_m) H_0 U_{H_m}(t_m) - H_0\|_{\operatorname{op}} \geq \frac 1 2\, \quad \text{with} \quad t_m \in [C_1 \lambda_m^{\frac{p}{\tau}},\  C_2 \lambda_m^{\frac{p}{\tau} + \epsilon}] \,,
       \end{equation}
       where $C_1 := \frac{\pi}{4}\left(\frac{\gamma}{2}\right)^{\frac{p}{\tau}}$, $C_2 := \frac{\pi}{4} |\nu|^{\frac{2 p}{\tau}} $ . 
\end{theorem}

\begin{remark}[Comments to Theorem \ref{thm:main}] \vphantom{ciao mamma}
\begin{enumerate}
    \item Theorem \ref{thm:main} holds in the thermodynamic limit $|\Lambda| \to +\infty$ and, due to Remark \ref{rmk:frecciarotta} and Definition \ref{def:PesciolinoDichiara}, 
    for a wide class of Hamiltonians its item \emph{(i)} relates two non-vanishing quantities in such limit.
    \item Item \emph{(i)} of Theorem \ref{thm:main} establishes the slow heating of the system. In particular, the quasi-conservation of $H_0$ implies the existence of a long-lived prethermal regime. If for instance $\tau = n$ and $b= \frac 12$, for any $\varepsilon:= \frac{n}{2p} \epsilon>0$ sufficiently small, it ensures that
	$$
	|\Lambda|^{-1}\| U_H^*(t) H_{0} U_H(t) - H_{0} \|_{\operatorname{op}} \leq D_1 \lambda^{-\frac 1 2} \quad \forall 0< t < \lambda^{\frac 1 2 (\frac{p}{ n} - 1) - \varepsilon}\,.
	$$
    \item 
    From a statistical mechanics point of view, Item \emph{(i)} of Theorem \ref{thm:main} means that the dynamics only spans a region for which $|\Lambda|^{-1}\|U_{H}^*(t) H_0 U_{H}(t) - H_0\|_{\operatorname{op}} \lesssim \lambda^{-b}$. Since thermalization would instead require $H_0$ to uniformly cover the entire accessible phase space, a necessary (but not sufficient) condition to observe thermalization would be to have $|\Lambda|^{-1} \|H_0(t) - H_0\|_{\operatorname{op}} \sim O(1)$. For this reason, the case $b \ll 1$ is particularly relevant, since the corresponding times represent a reasonable estimate of the lifetime of the prethermal state.
    \item Item \emph{(ii)} of Theorem \ref{thm:main} shows that the time-dependent dynamics generated by $H(\lambda \nu t)$ can be accurately approximated by the dynamics generated by an effective time-independent local Hamiltonian $H_\eff$. Such $H_\eff$ can be computed explicitly with an algorithmic procedure (see eq.~\eqref{eq:heff} below) and it is close to $H_0$.
    \item If, instead of a generic \(H\), one chooses $H(\lambda \nu t) = J \cdot N + V(\lambda \nu t)$, where $ J \cdot N = J_1 N_1 + \dots + J_r N_r$ is a linear combination of mutually commuting number operators $N_1, \dots, N_r$, then, using the same techniques and the non-resonance and locality assumptions of~\cite{Gallone-Langella-2024}, one can prove the quasi-conservation of all the number operators $N_j$ separately, within the same timescales found in Item \emph{(i)} of Theorem \ref{thm:main}. 
    \item As a consequence of the Theorem \ref{thm:main}, one has that if a perturbation is infinitely differentiable but not analytic, then the prethermal phase lasts longer than any polynomial in $\lambda$.
\end{enumerate}
\end{remark}

\begin{remark}[Comments to Theorem \ref{thm:fine.di.mondo}]\vphantom{ciao mamma}
    \begin{enumerate}
    \item Theorem \ref{thm:fine.di.mondo} establishes the optimality of the bound in Item \emph{(i)} of Theorem \ref{thm:main} for the limiting case $b=0$.  Indeed, one may wonder whether it is possible that Theorem \ref{thm:main} holds for any time quasi-periodic Hamiltonian $H \in \cO_{\kappa, C^p}$ of the form \eqref{eq:cosa.c.e.per.cena} on a longer timescale than the one stated in Item \emph{(i)}, namely for
    \begin{equation}\label{troppo.lungo}
      |t|\leq C'_2 \lambda^{\frac{p}{\tau} +\epsilon}\,, \quad \mbox{ for some } C'_2 >0 \quad \mbox{ and } 0 < \epsilon \ll 1\,.  
    \end{equation}
    Theorem \ref{thm:fine.di.mondo} proves that this is not possible, since in particular it exhibits a counterexample to the validity of Item \emph{(i)} of Theorem \ref{thm:main} on the longer timescales in \eqref{troppo.lungo} -- note to this aim that the threshold $\lambda_0$ in Theorem \ref{thm:main} depends on $H$ only through the quantities $\|H_0\|_{\kappa}$ and $\|V\|_{\kappa, C^p}$, and that such quantities in Theorem \ref{thm:fine.di.mondo} are uniformly bounded along the sequence $\{H_m\}_{m \in \N}$, see \eqref{sono.uniformi}.

    \item A related, and much more difficult, question with respect to the one tackled in Theorem \ref{thm:fine.di.mondo} is the one of exhibiting the existence of \emph{a single system} $H = H_0 + V(\lambda \nu t)$ for which one has
    \begin{equation}\label{non.item.i}
    |\Lambda|^{-1}\| U^*_H(t_m) H_0 U_H(t_m) -H_0\|_{\operatorname{op}} \sim O(1)
    \end{equation} on a sequence of times $t_m \to \infty$,
    instead of a collection of systems $\{H_m\}_{m \in \N}$ such that for any $m$ \eqref{non.item.i} is satisfied at $t_m$, with $t_m \to \infty$. This question remains widely open, and we believe it is beyond the reach of the techniques used in the present work.
        \item The same construction of Theorem \ref{thm:fine.di.mondo}  can also be adapted to prove optimality in the analytic setting. We point out that our perturbations are actually analytic, since they are given by a trigonometric polynomial (see eq.~\eqref{eq:big.mac} below). However, they are normalized in order to keep the $C^p$ norms uniformly bounded, whereas their analytic norms diverge as $k_m \to \infty$.
    \end{enumerate}
\end{remark}

\section{Stability estimates in the prethermal time scale}

\subsection{Analytic smoothing}
In this section we prove that a quasi local time quasi-periodic operator with finite smoothness  in time can be approximated by an analytic one. We follow the approach in \cite{Salamon2004}, which provides an analytic smoothing construction for scalar functions, and adapt it to operators -- the main delicate point in our analysis is to obtain estimates independent of the dimension of the local Hilbert space.

\begin{definition}\label{def:Onorato}
	Let $\sigma \geq0$, $\kappa \geq 0$ and let $\varphi \mapsto A(\varphi)$ be a family of operators defined by the interaction $\{A_S(\varphi)\}_{S \in \mathcal{P}_c(\Lambda)}$, with
    $\varphi \in \T^n$. 
    \begin{itemize}
    \item[i.] For any $S \in \mathcal{P}_c(\Lambda)$, we say that $A_S \in C^\omega_\sigma(\T^n; \mathcal{B}(\cH_S))$ if
    \begin{equation}
	\Vert A_S \Vert_{C^\omega_\sigma(\mathbb{T}^n; \mathcal{B}(\mathcal{H}_S))} := \sum_{\ell \in \mathbb{Z}^n} \Vert \widehat{A_S}(\ell) \Vert_{\mathrm{op}} e^{\sigma |\ell|} < +\infty \, .
\end{equation}
    \item[ii.] We say that $A \in \OOO_{\kkk,\sigma}$ if there exists a positive, finite constant $C$ independent of $\Lambda$ such that
	\begin{equation}
		\Vert A \Vert_{\kappa,\sigma} := \sup_{x \in \Lambda} \sum_{\substack{S \in \mathcal{P}_c(\Lambda) \\ x \in S}} \Vert A_S \Vert_{C^\omega_\sigma (\mathbb{T}^n; \mathcal{B}(\mathcal{H}_S))} e^{\kkk |S|} < C \, .
	\end{equation}
    \end{itemize}
\end{definition}

\begin{remark}
    Note that, if $A$ does not depend on $\varphi$, then $\Vert A \Vert_{\kappa,\sigma}=\Vert A \Vert_{\kappa,C^p}=\Vert A \Vert_\kappa$.
\end{remark}
The main result of this subsection is the following Proposition.
\begin{proposition}\label{prop.smooth.op}
	Let $\kappa \geq 0$, $p \geq n+1$, $\sigma \in (0,1)$, and $A \in \OOO_{\kkk,C^p}$. There exists $A_\sigma \in \OOO_{\kkk,\sigma}$ and two positive constants $\tC_1$ and $\tC_2$, depending only on $p$ and $n$, such that
\begin{itemize}
	\item[(i)] $\Vert A-A_\sigma \Vert_{\kkk,C^{0}} \leq \tC_1 \sigma^{p} \Vert A \Vert_{\kkk,C^p}$;
	\item[(ii)] $\Vert A_\sigma \Vert_{\kkk,\sigma} \leq \tC_2 \Vert A \Vert_{\kappa,C^p} $;
    \item[(iii)] $\langle A_\sigma \rangle = \langle A \rangle$\,, where $\langle A \rangle$ is defined as in \eqref{eq:mediamente}\,.
\end{itemize}
\end{proposition}

Let us first recall some standard facts that will be used in the following.
\begin{lemma}\label{lem:ComeCadeFourier}
    Let $A \in C^p(\mathbb{T}^n; \mathcal{B}(\cH_S))$, $p \geq n+1$. Then there exists $C_{n,p} >0$, depending on $n$ and $p$ only, such that for any $\ell \in \Z^n $
    \begin{equation}
        \Vert \widehat{A_S}(\ell) \Vert_{\operatorname{op}} \leq C_{n,p} \frac{\Vert A_S \Vert_{C^p(\mathbb{T}^n; \mathcal{B}(\cH_S))} }{\langle \ell\rangle ^p} \, .
    \end{equation}
\end{lemma}
\begin{proof}
    Let $\ell=(\ell_1,\dots,\ell_n) \in \mathbb{Z}^n$, and let $\ell_r$ be such that $|\ell_r|=\max_{j=1,\cdots,n} |\ell_j|$, then
    \[
          |\ell_r|^{p} \Vert \widehat{A_S}(\ell) \Vert_{\operatorname{op}} = \Big\Vert \frac{1}{(2 \pi)^n} \int_{\mathbb{T}^n}  e^{-\ii \varphi \cdot \ell} \partial_{\varphi_r}^p A_S(\varphi) \, d \varphi\Big\Vert_{\operatorname{op}} \leq \Vert A_S \Vert_{C^p(\mathbb{T}^n;\mathcal{B}(\cH_S))} \, .
    \]
    Using now that $|\ell_r|^p=\frac{1}{n^p}(n |\ell_r|)^p \geq \frac{1}{n^p} |\ell|^p \geq \frac{1}{(2n)^p} \langle \ell \rangle^p$ we obtain the thesis with $C_{n,p}=(2n)^p$.
\end{proof}

\begin{lemma}\label{lem:TaylorMan}
    Let $\vartheta,\varphi \in \mathbb{T}^n$ and $A_S \in C^p(\T^n;\mathcal{B}(\cH_S))$. Then, there exists $R_S^{(p)} \in C^0(\T^n \times \T^n;\mathcal{B}(\cH_S))$ such that
    \begin{equation}\label{eq:TaylorMcTaylor}
        A_S(\vartheta+\varphi)= A_S(\varphi) + \sum_{\substack{p' \in \mathbb{N}^n \\ 0< |p'| \leq p-1} }\frac{1}{p'!} \partial_\varphi^{p'}A_S(\varphi) \vartheta^{p'}+R_S^{(p)}(\varphi, \vartheta) \,,
    \end{equation}
    where $p'!=(p'_1!)\cdots (p'_n!)$, $\varphi^{p'}=\varphi_1^{p'_1}\cdots \varphi_n^{p'_n}$ and $\partial_\varphi^{p'}=\partial_{\varphi_1}^{p'_1}\cdots \partial_{\varphi_n}^{p'_n}$. Moreover, there exists $C'_{n,p}>0$ depending on $n$ and $p$ only such that
    \begin{equation}\label{eq:NorminaRR}
        \Vert R_S^{(p)}(\cdot, \vartheta) \Vert_{C^0(\T^n;\mathcal{B}(\cH_S))} \leq C'_{n,p} \Vert A_S \Vert_{C^p} |\vartheta|^p \quad \forall \vartheta \in \T^n\, .
    \end{equation}
\end{lemma}
\begin{proof}
    Using standard Taylor expansion one gets \eqref{eq:TaylorMcTaylor} with 
    \[
        R_S^{(p)}(\varphi,\vartheta)=\sum_{\substack{p' \in \mathbb{N}^n \\ |p'| = p}} \frac{|p|}{p'!}\, \vartheta^{p'} \int_0^1 (1-t)^{p-1} \partial_\varphi^{p'} A_S(\varphi+t \vartheta) \, dt \, .
    \]
    Then, from this latter explicit expression, we obtain the bound
    \[
        \begin{split}
            \Vert R_S^{(p)}(\cdot,\vartheta) \Vert_{C^0(\T^n;\mathcal{B}(\cH_S))} & = \sup_{\varphi \in \T^n} \Big\Vert \sum_{\substack{p' \in \mathbb{N}^n \\ |p'| = p}} \frac{|p|}{p'!} \int_0^1 (1-t)^{p-1} \partial_\varphi^{p'} A_S(\varphi+t \vartheta) \, dt \, \vartheta^{p'} \Big\Vert_{\operatorname{op}} \\
            &\leq \Big(\sum_{\substack{p' \in \mathbb{N}^n \\ |p'|=p}} \frac{|p|}{p!} \Big)\sup_{|p'| \leq p} \sup_{\varphi \in \T^n} \left\Vert 
            \int_0^1 (1-t)^{p-1} \partial_\varphi^{p'} A_S(\varphi+t \vartheta) \, dt \,  \vartheta^{p'}\right\Vert_{\operatorname{op}} \\
            &\hspace{-0.34cm}\overset{|1-t| \leq 1}{\leq} \Big(\sum_{\substack{p' \in \mathbb{N}^n \\ |p'|=p}} \frac{|p|}{p!} \Big) \sup_{\varphi \in \mathbb{T}^n} \sup_{|p'| \leq p} \Vert \partial_\varphi^{p'}A_S(\varphi) \Vert_{\operatorname{op}} (\max_{j=1,\dots,n} |\vartheta_j|)^p\,.
        \end{split}
    \]
Now, using that $\Vert A_S \Vert_{C^p(\T^n;\mathcal{B}(\cH_S))}=\sup_{|p'| \leq p} \Vert \partial_\varphi^{p'} A_S(\varphi) \Vert_{\operatorname{op}}$ we get the thesis.
\end{proof}
We can now define
\begin{equation}\label{eq:PoliTaylor}
    T_{A_S}^{(p-1)}(\varphi,\vartheta):=\sum_{\substack{p' \in \mathbb{N}^n \\ 0< |p'| \leq p-1}} \frac{1}{p'!} \partial_\varphi^{p'} A_S(\varphi) \vartheta^{p'} \, .
\end{equation}
Then, as a consequence of \eqref{eq:TaylorMcTaylor} and \eqref{eq:NorminaRR}, for any $\varphi, \vartheta \in \T^n$ one has
\begin{equation}
    \Vert A_S(\varphi+\vartheta) -A_S(\varphi)+T_{A_S}^{(p-1)}(\varphi,\vartheta) \Vert_{\operatorname{op}} = \Vert R^{(p)}_S(\varphi,\vartheta) \Vert_{\operatorname{op}} \leq C'_{n,p} |\vartheta|^p \Vert A_S \Vert_{C^p(\T^n;\mathcal{B}(\cH_S))} \, .
\end{equation}
With this last equation, we completed the list of properties of $A_S$ that we require. We now move to discussing the construction of the analytic smoothing, which is done by convolution with a suitable kernel.

Let us denote by $B_r(0):=\{x \in \mathbb{R}^n \, | \, |x| < r\}$ the ball of radius $r$ in $\mathbb{R}^n$ centered at $0$ and let us consider a function $\chi \in C^\infty_c(\mathbb{R}^n)$ such that
\begin{itemize}
    \item[(i)] $\mathrm{supp} \, \chi \subset B_1(0)$;
    \item[(ii)] $\chi(\xi)=1$ for $\xi \in B_{\frac{1}{2}}(0)$.
\end{itemize}
We define the kernel
\begin{equation}\label{eq:k.kernel}
    K(z):= \frac{1}{(2 \pi)^n }\int_{\mathbb{R}^n} \chi(\xi) e^{\ii \xi \cdot z} \, \ud \xi\,, \quad z \in \R^n\, .
\end{equation}
Note that $K(z)$ is a Schwartz function, since it is the Fourier transform of a Schwartz function.
The following result can be found in \cite[Lemma 9]{Chierchia-LN}, but we include an adapted proof here for completeness.

\begin{lemma} \label{lem:KeKapponeKarino}
    Let $K$ be the kernel defined in \eqref{eq:k.kernel}. Then
    \begin{itemize}
        \item[(i)] $\int_{\mathbb{R}^n} K(z) \, dz = 1$;
        \item[(ii)] $\int_{\mathbb{R}^n} K(z) z^{p'} \, dz=0$ if $p' \neq 0$, $p' \in \mathbb{N}^n$.
    \end{itemize}
\end{lemma}
\begin{proof}
Item (i) follows from $\int_{\mathbb{R}^n} K(z) \, dz = \chi(0)=1$.

Concerning item (ii),
\[
    \begin{split}
        \int_{\mathbb{R}^n} K(z) z^{p'} \, dz &= \int_{\mathbb{R}^n} K(z) z^{p'} e^{-\ii \xi \cdot z} \, dz \Big|_{\xi=0} =\ii^{|p'|} \int_{\mathbb{R}^n} K(z) \partial_\xi^{p'} e^{-\ii \xi \cdot z} \, dz \Big|_{\xi=0} \\
        &=\ii^{|p'|} \partial_{\xi}^{p'} \int_{\mathbb{R}^n} K(z) e^{-\ii \xi \cdot z} \, dz \Big|_{\xi=0} =\ii^{|p'|} \partial_\xi^{p'}\chi(\xi)\big|_{\xi=0}=0 \, ,
    \end{split}
\]
where in the last step we used the fact that $\chi$ is constant around zero.
\end{proof}

\begin{lemma}\label{lem:Casseruola}
	Let $p \geq n+1$, $\sigma \in (0,1)$, and $A_S \in C^p(\mathbb{T}^n; \mathcal{B}(\mathcal{H}_S))$. There exist $  (A_S)_\sigma \in C^\omega(\mathbb{T}^n_\sigma;\mathcal{B}(\mathcal{H}_S))$, and positive constants $\tC_1$ and $\tC_2$ depending only on $n$ and $p$ such that
	\begin{itemize}
		\item[(i)] $\Vert A_S-(A_S)_\sigma \Vert_{C^{0}(\mathbb{T}^n;\mathcal{B}(\mathcal{H}_S))} \leq \tC_1 \sigma^{p} \Vert A_S \Vert_{C^p(\mathbb{T}^n;\mathcal{B}(\mathcal{H}_S))}$,
		\item[(ii)] $\Vert (A_S)_\sigma \Vert_{C^\omega_\sigma(\mathbb{T}^n_\sigma;\mathcal{B}(\mathcal{H}_S))} \leq \tC_2  \Vert A_S \Vert_{C^p(\mathbb{T}^n;\mathcal{B}(\mathcal{H}_S))}$,
        \item[(iii)] $\langle (A_S)_\sigma \rangle = \langle A_S \rangle$.
	\end{itemize}
\end{lemma}

\begin{proof}
    We define $(A_S)_{\sigma}(\varphi):= \int_{\mathbb{R}^n} K(z) A_S(\varphi-\sigma z) \, dz$. By direct inspection, $(A_S)_{\sigma}$ is $2 \pi$ periodic in each of its component and one has $\widehat{(A_S)_{\sigma}}(\ell)=\chi(\sigma \ell) \widehat{A_S}(\ell)$ for any $\ell \in \Z^n$. This immediately implies that $\langle A_S \rangle=\langle (A_S)_\sigma \rangle$, which is Item (iii). We now prove Item (ii): one has
        \[
        \begin{split}
            \Vert (A_S)_\sigma \Vert_{C^\omega_\sigma (\T^n;\mathcal{B}(\cH_S))} &= \sum_{\ell \in \mathbb{Z}^n} \Vert \widehat{A_S}(\ell) \Vert_{\operatorname{op}} \chi(\sigma \ell) e^{\kappa |\ell|} \\
            &\hspace{-0.45cm}\overset{ \text{Lemma \ref{lem:ComeCadeFourier}}}{\leq} C_{n,p} \sum_{\ell \in \mathbb{Z}^n} \Vert A_S \Vert_{C^p(\T^n,\mathcal{B}(\cH_S))} \frac{1}{\langle \ell\rangle ^p} \chi(\sigma \ell) e^{\sigma |\ell|}\,.
        \end{split}
    \]
    Since  $p\geq n+1$ and $\chi(\sigma \ell) \neq 0$ only for $\sigma |\ell| \leq 1$, the last series is convergent and its value depends on $p$ and $n$ only, thus we have proved (ii).

    Concerning (i), we now first use item (i) in Lemma \ref{lem:KeKapponeKarino} and write $A_S(\varphi)=A_S(\varphi) \int_{\mathbb{R}^n} K(z) \, dz$. Then
    \[
        \begin{split}
            A_S(\varphi)-(A_S)_\sigma(\varphi)&= \int_{\mathbb{R}^n} K(z) (A_S(\varphi)-A_S(\varphi-\sigma z)) \, dz \\ &\hspace{-0.42cm}\overset{\eqref{eq:TaylorMcTaylor},\eqref{eq:PoliTaylor}}= \int_{\mathbb{R}^n} K(z) \big( T_{A_S}^{(p-1)}(\varphi,-\sigma z)-R_S^{(p)}(\varphi,-\sigma z)\big) \, dz \\
            &\hspace{-0.69cm}\overset{\text{Lemma \ref{lem:KeKapponeKarino}-(ii)}}{=} -\int_{\mathbb{R}^n} K(z) R_S^{(p)}(\varphi,-\sigma z) \, \ud z\,,
        \end{split}
    \]
    where in the last step we used the fact that $T_{A_S}^{(p-1)}$ is a polynomial in $z$ and that $K$ integrated against any polynomial is zero.
    We now estimate both sides of the last equation in $C^0$ norm to get
    \[
        \begin{split}
        \Vert A_S - (A_S)_\sigma \Vert_{C^0(\T^n;\mathcal{B}(\cH_S))} &= \sup_{\varphi \in \T^n} \left\Vert \int_{\mathbb{R}^n} K(z) R_S^{(p)}(\varphi,-\sigma z) \, \ud z \right \Vert_{\operatorname{op}} \\ 
        &\leq C_{p,n} \sigma^p \Vert A_S \Vert_{C^p(\T^n;\mathcal{B}(\cH_S))} \int_{\mathbb{R}^n} |K(z)| |z|^p \, dz \leq  C_{p,n} \sigma^p \Vert A_S \Vert_{C^p(\T^n; \cB(\cH_S))}\,.
        \end{split}
    \]
\end{proof}

\begin{proof}[Proof of Proposition \ref{prop.smooth.op}]
	For any $S \in \mathcal{P}_c(\Lambda)$ one constructs $(A_S)_\sigma$ as in Lemma \ref{lem:Casseruola}. We define $A_\sigma:=\sum_{S \in \mathcal{P}_c(\Lambda)} (A_S)_\sigma$ and we show that such $A_\sigma$ satisfies items (i)-- (iii). Item (i) is proven using Lemma \ref{lem:Casseruola}
	\begin{equation}
		\begin{split}
			\Vert A - A_\sigma\Vert_{\kkk,C^{0}} &=\sup_{x \in \Lambda} \sum_{\substack{S \in \mathcal{P}_c(\Lambda) \\ x \in S}} \Vert A_S-(A_S)_\sigma \Vert_{C^{0}(\mathbb{T}^n; \mathcal{B}(\mathcal{H}_S))} e^{\kappa|S|} \\
			&\hspace{-0.7cm}\overset{\text{Lemma \ref{lem:Casseruola}-(i)}}{\leq} \sup_{x \in \Lambda} \sum_{\substack{S \in \mathcal{P}_c(\Lambda) \\ x \in S}} \tC_1 \sigma^{p} \Vert A_S \Vert_{C^p(\mathbb{T}^n; \mathcal{B}(\mathcal{H}_S))} e^{\kkk |S|} \\
			&= \tC_1 \sigma^{p}\sup_{x \in \Lambda} \sum_{\substack{S \in \mathcal{P}_c(\Lambda) \\ x \in S}} \Vert A_S \Vert_{C^p(\mathbb{T}^n;\mathcal{B}(\mathcal{H}_S))} e^{\kkk |S|} = \tC_1 \sigma^{p} \Vert A \Vert_{\kappa,C^p} \, .
		\end{split}
	\end{equation}
	Concerning the second item,
	\begin{equation}
		\begin{split}
			\Vert A_\sigma \Vert_{\kkk, \sigma} &= \sup_{x \in \Lambda} \sum_{\substack{S \in \mathcal{P}_c(\Lambda) \\ x \in S}} \Vert( A_S)_\sigma \Vert_{C^\omega_\sigma(\mathbb{T}^n_\sigma;\mathcal{B}(\mathcal{H}_S))} e^{\kkk |S|} \\
			&\hspace{-0.75cm}\overset{\text{Lemma \ref{lem:Casseruola}-(ii)}}{\leq} \sup_{x \in \Lambda} \sum_{\substack{S \in \mathcal{P}_c(\Lambda) \\ x \in S}} \tC_2 \Vert A_S \Vert_{C^p(\mathbb{T}^n;\mathcal{B}(\mathcal{H}_S))} e^{\kkk|S|} = \tC_2 \Vert A \Vert_{\kkk,C^p} \, .
		\end{split}
	\end{equation}
    Item (iii) follows immediately by linearity from Item (iii) of Lemma \ref{lem:Casseruola}.
\end{proof}

\subsection{Time rescaling and normal form}\label{sec:nf} 
Given a time dependent Hamiltonian $H$, let $U_H(t,s)$ be the 2-parameters family of unitary operators solving the Cauchy problem
 \begin{equation}\label{flow.H}
 \begin{aligned}
 &\partial_t U_H(t, s) = -\ii H(\lambda \nu t) U_H(t,s)\,, \\
 &U_H(s,s) = \mathbb{1}\,.
 \end{aligned}
 \end{equation}
 Equivalently, $\tU_H(\tit, \ts):= U_H(\lambda^{-1} \tit, \lambda^{-1} \ts)$ satisfies
 \begin{equation}\label{flow.tU}
 \begin{aligned}
 &\partial_{\tit} \tU_H(\tit, \ts) = -\ii \lambda^{-1} H(\nu \tit) \tU_H(\tit,\ts)\,, \\
 &\tU_H(\ts,\ts) = \mathbb{1}\,.
 \end{aligned}
 \end{equation}
 In all cases, we denote $U_H(t,0) \equiv U_H(t)$, $\tU_H(\tit, 0) \equiv \tU_H(\tit)$. 
 
 For the iterative procedure we will rely on an analytic smoothing of $V$ in which the analyticity strip $\sigma$ depends on $\lambda$ (see \eqref{lolly.molly.dolly} below).
 An immediate consequence of Proposition \ref{prop.smooth.op}
 is the following: 
 \begin{corollary}[Corollary to Proposition \ref{prop.smooth.op}]\label{cor.lissage}
 	Let $n \in \N$, $p \geq n + 1$ and $\kappa \geq 0$. There exist $\tC_1, \tC_2>0$ depending on $p$ and $n$ only, $V_\lambda \in \cO_{\kappa, \sigma}$, and $R_\lambda \in \cO_{\kappa, C^{0}}$ such that
 	\begin{equation}
 	H(\nu \tit) = H_0 + V_{\lambda}(\nu \tit) + R_\lambda(\nu \tit)\,,
 	\end{equation}
 	with $\|V_\lambda \|_{\kappa
    , \sigma} \leq  \tC_1 \|V\|_{\kappa, C^p}$ and $\|R_\lambda\|_{\kappa, C^0} \leq \tC_2 \sigma^p \|V\|_{\kappa, C^p}\,.$
 \end{corollary}
 
 Consistently with the notation of Proposition \ref{prop.smooth.op}, we could have used $V_\sigma$ and $R_\sigma$  instead of $V_\lambda$ and $R_\lambda$ but since the two parameters are linked by \eqref{lolly.molly.dolly}, we prefer to use $V_\lambda$ and $R_\lambda$ in the following.
 
 \begin{definition}[Conjugate dynamics]
 	Given $H_1(\nu \tit)$ and $H_2(\nu \tit)$ two time-quasiperiodic self-adjoint operators, we say that the time quasi-periodic unitary operator $Y(\nu \tit)$ conjugates the dynamics of $H_1(\nu \tit)$ to the dynamics of $H_2(\nu \tit)$ if
 	\begin{equation}\label{def.conj}
 	U_{H_2}(\tit)= Y(\nu \tit) U_{H_1}(\tit) \,.
 	\end{equation}
 \end{definition}
The goal of this section is to prove the following: 
 \begin{proposition}[Normal form]\label{normal.form}
 	Let $	\kappa_{\fin} := \frac{\kappa}{2}$ and fix $b\in (0,\ \frac{p}{p + \tau})$ and $\epsilon \in (0, 1 - \frac{p+\tau}{p} b)$. Let furthermore
 	\begin{equation}\label{lolly.molly.dolly}
	\sigma := \lambda^{-\tt A}\,, \quad {\tt A}:=\frac{1-b- \epsilon}{\tau} \,.
 	\end{equation}
 	There exists an intensive constant $\lambda_0$, such that, if $\lambda>\lambda_0$, the following holds. There exist
 	a unitary transformation $Y^{(\fin)}$ which conjugates the dynamics of $\lambda^{-1} H$ to the dynamics of a self-adjoint operator $H^{(\fin)} \in \cO_{\kappa_\fin, 0}$,
 	with the following properties:
 	\begin{equation}\label{eq:Hfin}
 	H^{(\fin)}(\nu \tit) = \lambda^{-1} \left(H_0 + Z_\lambda^{(\fin)} + V_{\lambda}^{(\fin)}(\nu \tit) +  R_{\lambda}^{(\fin)} (\nu \tit) \right)\,,
 	\end{equation}
 	with
 	\begin{enumerate}
 	\item $Z_\lambda^{(\fin)}$ time independent, $\|Z_\lambda^{(\fin)}\|_{\kappa_\fin} \leq D \lambda^{-b} \|V\|_{\kappa, C^p}$
 	\item $\| V_\lambda^{(\fin)}\|_{\kappa_\fin, 0} + \|R_\lambda^{(\fin)}\|_{\kappa_\fin, C^0} \leq D \lambda^{-\frac{p(1-b-\epsilon)}{\tau}} \|V\|_{\kappa, C^p}\,,$ $\quad D :=  \frac{2 e}{e-1} \max\{\tC_1, \tC_2\}\,,$ with $\tC_1$, $\tC_2$ as in Corollary \ref{cor.lissage}
 	\item $\| Y(\nu \tit) A Y^*(\nu \tit) - A\|_{\kappa_\fin, C^0} \leq \frac{2e}{e-1}\lambda^{-b} \|A\|_{\kappa_\fin, C^0}$ \quad for any $A \in \cO_{\kappa_\fin, C^0}$.
 	\end{enumerate}
 \end{proposition}

Before proving the Lemma, we state some intermediate auxiliary results.
\begin{lemma}\label{lemma.porca.troia}
	For any $A \in \cO_{\kappa, 0}$, with $\kappa \geq 0$,  we have
	$\|A\|_{\kappa, C^0} \leq \|A\|_{\kappa, 0}\,.$ 
\end{lemma}
\begin{proof}
	For any $S \in \mathcal{P}_c(\Lambda)$ one has
	$$
	\sup_{\varphi \in \T^n}\|A_S(\varphi)\|_{\oop} = \sup_{\varphi \in \T^n} \left\| \sum_{\ell \in \Z^n} \widehat{A_S}(\ell) e^{\ii \ell \cdot \varphi}\right\|_{\oop} \leq \sum_{\ell \in \Z^n} \| \widehat{A_S}(\ell)\|_{\oop}\,,
	$$
	thus
	$$
	\|A\|_{\kappa, C^0} = \sup_{x \in \Lambda} \sum_{\substack{S \in \mathcal{P}_c(\Lambda) \\ x \in S}} \sup_{\varphi \in \T^n}\|A_S(\varphi)\|_{\oop} e^{\kappa|S|} \leq \sup_{x \in \Lambda} \sum_{\substack{S \in \mathcal{P}_c(\Lambda) \\ x \in S}}  \sum_{\ell \in \Z^n} \| \widehat{A_S}(\ell)\|_{\oop} e^{\kappa |S|} = \|A\|_{\kappa, 0}\,.
	$$
\end{proof}

In the next Lemma, to show smallness of $e^{G}Ae^{-G}-A$ we rely on the representation
\begin{equation}\label{eq:commuta.che.ti.passa}
    e^G A e^{-G}=\sum_{r=0}^{+\infty} \frac{1}{r!} \mathrm{Ad}^r_G A \, , \qquad \mathrm{Ad}^r_G A:=[G,\mathrm{Ad}^{r-1}_G A] \quad \forall r \geq 1 \, , \qquad \mathrm{Ad}^0_G A=A \, .
\end{equation}

\begin{lemma}[Commutator expansions]\label{lemma.commut}
	Let $\delta >0$ and $G \in \cO_{\kappa, \sigma}$ for some $\kappa, \sigma \geq 0$. If there exists $\eta \in (0, 1)$ such that
	\begin{equation}\label{little.power.series}
	\frac{4 e^{-\kappa} \|G\|_{\kappa + \delta, \sigma}}{\delta} \leq \eta\,,
	\end{equation}
	then
	\begin{itemize}
	\item[(i)] For any $A \in \cO_{\kappa + \delta, \sigma}$ one has
	\begin{align}
	\|e^{G} A e^{-G} - A\|_{\kappa, \sigma} &\leq \frac{C_\eta e^{-\kappa}}{\delta} \|A\|_{\kappa + \delta, \sigma} \|G\|_{\kappa + \delta, \sigma} 
	\,,
	\end{align}
	with $C_\eta := \frac{4}{1-\eta}>0$.
	\item[(ii)] For any $A \in \cO_{\kappa + \delta, C^0}$ one has
	\begin{align}
	\|e^{G} A e^{-G} - A\|_{\kappa, C^0} &\leq \frac{C_\eta e^{-\kappa}}{\delta} \|A\|_{\kappa + \delta, C^0} \|G\|_{\kappa + \delta, 0} 
	\,.
	\end{align}
	\end{itemize}
\end{lemma}
\begin{proof}
	Item $(i)$ is proven in \cite[Lemma 4.2]{Gallone-Langella-2024}. The proof of Item $(ii)$ is analogous, but here we repeat it for the sake of completeness. We start with proving that, for any $ \uozokozo  \in \N$,
		\begin{equation}\label{ad.daveni}
		\Vert \textnormal{Ad}^ \uozokozo _G A \Vert_{\kappa, C^0} \; \leq \; \left(\frac{ \uozokozo }{e}\right)^ \uozokozo  \frac{(4 e^{-\kappa})^ \uozokozo }{\delta^{ \uozokozo }} \Vert G\Vert^{ \uozokozo }_{\kappa + \delta, C^0} \Vert A \Vert_{\kappa + \delta, C^0}\,.
		\end{equation}	
        Indeed, $\forall  \uozokozo $ one has
		\[
		\begin{split}
		\Vert \mathrm{Ad}^ \uozokozo _G A \Vert_{\kappa, C^0} \;&=\; \sup_{x \in \Lambda} \sum_{\substack{S \in \mathcal{P}_c(\Lambda) \\ x \in S}} e^{\kappa |S|}  \sup_{\varphi \in \T^n} \Vert (\mathrm{Ad}_G^ \uozokozo  A)_S (\varphi) \Vert_{\operatorname{op}} \\
		&=\;\sup_{x \in \Lambda} \sum_{\substack{S \in \mathcal{P}_c(\Lambda) \\ x \in S}} \sum_{\substack{S_1,\dots,S_{ \uozokozo +1} \in \mathcal{P}_c(\Lambda) \, \text{s.t.} \\ \bigcup_{j=1}^{ \uozokozo +1} S_j = S \\ S_r \cap (\bigcup_{r+1}^{ \uozokozo +1} S_\ell) \neq \varnothing \\ \forall r=1, \dots,  \uozokozo }}\sum_{k \in \mathbb{Z}^n} e^{\kappa |S|} \tilde{F}_{S_1,\dots,S_{ \uozokozo +1}}(A,G)\,,
		\end{split}
		\]
		with
		\[
		\tilde{F}_{S_1,\dots,S_{ \uozokozo +1}}(A,G)\;:=\; \sup_{\varphi \in \T^n} \Big \Vert  \Big[{G_{S_1}}(\varphi),\Big[{G_{S_2}}(\varphi),\Big[\dots,\Big[{G_{S_ \uozokozo }}(\varphi),{A_{S_{ \uozokozo +1}}}(\varphi)\Big] \dots\Big]\Big]\Big] \Big\Vert_{\operatorname{op}}\,.
		\]
		Using the fact that $e^{\kappa|S|} \leq  e^{\kappa(|S_1| + \dots + |S_ \uozokozo | -  \uozokozo )}$ and that
		$$
		\tilde{F}_{S_1,\dots,S_{ \uozokozo +1}}(A,G)\leq 2^{ \uozokozo } F_{S_1, \dots, S_{ \uozokozo +1}} (A, G)\,,
		$$
		with
		$$
		\begin{aligned}
		&F_{S_1, \dots, S_{ \uozokozo +1}} (A, G) := \sup_{\varphi \in \T^n} \Vert {G_{S_1}}(\varphi) \Vert_{\operatorname{op}}  \cdots \Vert {G_{S_{ \uozokozo }}}(\varphi)	\Vert_{\operatorname{op}} \Vert {A_{S_{ \uozokozo +1}}} (\varphi)\Vert_{\operatorname{op}}\\
		&\leq \sup_{\varphi \in \T^n} \Vert {G_{S_1}}(\varphi) \Vert_{\operatorname{op}} \cdots \sup_{\varphi \in \T^n} \Vert {G_{S_ \uozokozo }}(\varphi) \Vert_{\operatorname{op}} \sup_{\varphi \in \T^n} \Vert {A_{S_{ \uozokozo +1}}}(\varphi) \Vert_{\operatorname{op}}\,,
		\end{aligned}
		$$
		we get that
		\begin{align*}
		\Vert \textnormal{Ad}_G^ \uozokozo (A)\Vert_{\kappa, C^0} & \leq (2 e^{-\kappa})^{ \uozokozo } \sup_{x \in \Lambda} \sum_{\substack{S_1,\cdots,S_{r+1} \in \mathcal{P}_c(\Lambda) \, \text{s.t.} \\ x \in S_1 \cup \dots \cup S_{ \uozokozo +1} \\
				S_r \cap (\bigcup_{s={r+1}}^{ \uozokozo +1} S_s) \neq \varnothing \\ \forall r=1,\dots, \uozokozo }} e^{(|S_1| + \dots + |S_{ \uozokozo +1}|)\kappa}  F_{S_1, \dots, S_{ \uozokozo +1}} (A, G)\,.
		\end{align*}	
		Then by Lemma A.1 of \cite{Gallone-Langella-2024} one obtains
		\begin{align*}
		\Vert \textnormal{Ad}^ \uozokozo _G A \Vert_{\kappa, C^0} &\leq  (4 e^{-\kappa})^ \uozokozo  \max_{\sigma \in \mathfrak{S}_{ \uozokozo +1}} \sup_{x_{\sigma(1)} \in \Lambda} \sum_{\substack{S_1 \in \mathcal{P}_c(\Lambda) \\ x_{\sigma(1)} \in S_1}} \dots \sup_{x_{\sigma( \uozokozo +1)} \in \Lambda} \sum_{\substack{S_{ \uozokozo +1} \in \mathcal{P}_c(\Lambda) \\ x_{\sigma(\uozokozo+1)} \in S_{r+1}}} e^{\kappa(|S_1| + \dots |S_{ \uozokozo +1}|)}\cdot \\
		& \qquad \quad \cdot (|S_1| + \dots + |S_{ \uozokozo +1}|)^{ \uozokozo } F_{S_1, \dots, S_{ \uozokozo +1}}(A, G)\,,
		\end{align*}
        where $\mathfrak{S}_{r+1}$ is the set of permutations of $r+1$ elements. By Cauchy estimate $\max_{x \in \R^+} e^{-\delta x} x^{ \uozokozo } = \left(\frac{ \uozokozo }{e}\right)^ \uozokozo  \frac{1}{\delta^ \uozokozo }$ one gets
        $$
		\begin{aligned}
		\Vert \textnormal{Ad}^ \uozokozo _G A \Vert_{\kappa, C^0} &\leq (4 e^{-\kappa})^ \uozokozo  \left(\frac{ \uozokozo }{e}\right)^ \uozokozo \delta^{- \uozokozo } \cdot \\
		&\qquad \cdot \max_{\sigma \in \mathfrak{S}_ \uozokozo } \sup_{x_{\sigma(1)} \in \Lambda} \sum_{\substack{S_{1} \in \mathcal{P}_c(\Lambda) \\ x_{\sigma(1)} \in S_1}} \dots \sup_{x_{\sigma( \uozokozo )} \in \Lambda} \sum_{\substack{S_{ \uozokozo +1} \in \mathcal{P}_c(\Lambda) \\ x_{\sigma( \uozokozo +1)} \in S_{ \uozokozo +1}}}  e^{(\kappa +\delta)(|S_1| + \dots |S_{ \uozokozo +1}|)}  F_{S_1, \dots, S_{ \uozokozo +1}}(A, G)\\
		& = (4 e^{-\kappa})^ \uozokozo  \left(\frac{ \uozokozo }{e}\right)^ \uozokozo \delta^{- \uozokozo } \Vert G\Vert_{\kappa + \delta, C^0}^{ \uozokozo } \Vert A\Vert_{\kappa + \delta, C^0}\,.
		\end{aligned}
		$$
		Then, by \eqref{ad.daveni}, one has
		\begin{align*}
		\Big\Vert e^{- G} A e^{ G} - A \Big\Vert_{\kappa, C^0}& = \Big\Vert \sum_{ \uozokozo  \geq 1} \frac{1}{ \uozokozo !}\textnormal{Ad}^ \uozokozo _G A \Big\Vert_{\kappa,C^0} \leq \sum_{ \uozokozo \geq 1} \frac{1}{ \uozokozo !}\left(\frac{ \uozokozo }{e}\right)^ \uozokozo  \frac{(4 e^{-\kappa})^ \uozokozo }{\delta^{ \uozokozo }}\Vert G\Vert^{ \uozokozo }_{\kappa + \delta, C^0} \Vert A \Vert_{\kappa + \delta, C^0} \\ &= \sum_{ \uozokozo  \geq 1} \frac{1}{ \uozokozo !} \left(\frac{ \uozokozo }{e} \right)^ \uozokozo  \left(\frac{4 e^{-\kappa} \|G\|_{\kappa + \delta,C^0}}{\delta}\right)^ \uozokozo  \|A\|_{\kappa + \delta,C^0}\\
		&= \frac{4 e^{-\kappa} \|G\|_{\kappa + \delta,C^0}}{\delta} \sum_{ \uozokozo  \geq 0} \frac{1}{( \uozokozo +1)!} \left(\frac{ \uozokozo +1}{e} \right)^{ \uozokozo +1} \left(\frac{4 e^{-\kappa} \|G\|_{\kappa + \delta, C^0}}{\delta}\right)^ \uozokozo  \|A\|_{\kappa + \delta, C^0}\,.
		\end{align*}
		Since by Stirling formula $ \uozokozo ! \geq \left(\frac{ \uozokozo }{e}\right)^ \uozokozo  \sqrt{2\pi  \uozokozo }$, one has
		$\frac{1}{( \uozokozo +1)!} \left(\frac{ \uozokozo +1}{e}\right)^{ \uozokozo +1} \leq 1$, thus recalling Lemma \ref{lemma.porca.troia} and the smallness assumption \eqref{little.power.series}, one deduces
		$$
		\Big\Vert e^{- G} A e^{ G} - A \Big\Vert_{\kappa, C^0} \leq \frac{4 e^{-\kappa} \|G\|_{\kappa + \delta,C^0}}{\delta} \|A\|_{\kappa + \delta, C^0} \sum_{ \uozokozo \geq 0} \eta^{ \uozokozo } \leq  \frac{4 e^{-\kappa} \|G\|_{\kappa + \delta,C^0}}{\delta(1-\eta)} \|A\|_{\kappa + \delta, C^0}\,.
		$$
	\end{proof}
\begin{remark}\label{vattelapesca}
	Given $A \in \cO_{\kappa, \sigma}$, then its time average $\langle A \rangle$ defined in \eqref{eq:mediamente} satisfies $\| \langle A \rangle \|_{\kappa} = \|\langle A \rangle\|_{\kappa, \sigma} \leq \|A\|_{\kappa, \sigma}$.
\end{remark}

\begin{definition}\label{eq:ServirebbeLaCremaSolare}
	Let $A \in \cO_{\kkk,\sigma}$ and fix $K>0$. We define $A^\uv$ and $A^\ir$ as
	\begin{equation}\label{eq:ir.uv}
		A^\ir(\varphi)=\sum_{S \in \mathcal{P}_c(\Lambda)} \sum_{\substack{\ell \in \mathbb{Z}^n \\ |\ell| \leq K}} \widehat{A_S}(\ell) e^{\ii \ell \cdot \varphi} \, , \qquad A^\uv(\varphi)=A(\varphi)-A^\ir(\varphi) \, .
	\end{equation}
\end{definition}
In the following, we shall not mention explicitly the dependence of $A^\uv$ and $A^\ir$ on $K$.

We recall the following immediate result:
\begin{lemma}[Ultraviolet cut-off]\label{lemma.uv}
	If $A \in \cO_{\kappa, \sigma}$ for some $\kappa\geq 0$ and $\sigma >0$, then $\|A^\ir\|_{\kappa, \sigma}\leq \|A\|_{\kappa, \sigma}$, $ \| A^\ir - \langle A \rangle \|_{\kappa, \sigma} \leq \|A\|_{\kappa, \sigma}$, and
	$\|A^\uv\|_{\kappa, 0} \leq e^{-K \sigma} \|A\|_{\kappa, \sigma}\,.$
\end{lemma}
\begin{proof}
We only prove the statement for $A^\uv$, as the ones for $A^\ir$ and $A^\ir - \langle A \rangle$ can be deduced in an analogous way. One has
\begin{align*}
    \| A^\uv\|_{\kappa, 0} &= \sup_{x \in \Lambda} \sum_{\substack{S \in \mathcal{P}_c(\Lambda) \\ x \in S}} \sum_{\ell \in \Z^n \atop |\ell| > K} \|\widehat{A_S}(\ell)\|_{\operatorname{op}} e^{\kappa |S|} \leq \sup_{x \in \Lambda} \sum_{\substack{S \in \mathcal{P}_c(\Lambda) \\ x \in S}} \sum_{\ell \in \Z^n \atop |\ell| > K} \|\widehat{A_S}(\ell)\|_{\operatorname{op}} e^{\kappa |S|} e^{-K \sigma} e^{|\ell| \sigma}= e^{-K \sigma} \|A\|_{\kappa, \sigma}\,.
\end{align*}
\end{proof}
\begin{lemma}\label{lemma.homo}
	Let $A \in \cO_{\kappa, \sigma}$, $\nu \in DC^n(\gamma,\tau)$ and $K>0$. The equation
	\begin{equation}\label{hom.sweet.hom}
	(\nu \cdot \partial_{\varphi}) G(\varphi) + A(\varphi) = A^{\uv}(\varphi) + \langle A \rangle
	\end{equation}
	admits a solution $G \in \cO_{\kappa, \sigma}$ with
	$\|G\|_{\kappa, \sigma} \leq \gamma^{-1} K^\tau\|A\|_{\kappa, \sigma}\,.$ 
\end{lemma}
\begin{proof}
	Let
	$$
	G(\varphi):= \sum_{S \in \mathcal{P}_c(\Lambda)} \sum_{\ell \in \Z^n\ \atop 0<|\ell| \leq K} \widehat{G_S}(\ell) e^{\ii \ell \cdot \varphi}\,, \quad \widehat{G_S}(\ell) := -\frac{ \widehat{A_S}(\ell)}{\ii \nu \cdot \ell} \quad \forall \ell \in \mathbb{Z}^n\,, 0<|\ell| \leq K.
	$$
	Then $\nu\cdot \partial_{\varphi} G(\varphi) = - A^{\ir}(\varphi) + \langle A \rangle$, thus, recalling $A(\varphi) = A^{\ir}(\varphi) + A^\uv(\varphi)$, $G(\varphi)$ solves \eqref{hom.sweet.hom}. Furthermore, since $\nu \in DC^n(\gamma, \tau)$, one has
	$$
    \begin{aligned}
       \|G\|_{\kappa, \sigma} &= \sup_{x \in \Lambda} \sum_{\substack{S \in \mathcal{P}_c(\Lambda) \\ x \in S}} \sum_{\ell \in \Z^n \atop 0<|\ell|\leq K} \frac{\|\widehat{A_S}(\ell)\|_{\oop}}{|\nu \cdot \ell|} e^{\kappa|S|} e^{|\ell|\sigma} \\
       &\leq \gamma^{-1} K^\tau \sup_{x \in \Lambda} \sum_{\substack{S \in \mathcal{P}_c(\Lambda) \\ x \in S}} \sum_{\ell \in \Z^n \atop 0<|\ell|\leq K} \|\widehat{A_S}(\ell)\|_{\oop} e^{\kappa|S|} e^{|\ell|\sigma}
       \leq \gamma^{-1} K^\tau \|A\|_{\kappa, \sigma}\,. 
    \end{aligned}
	$$
\end{proof}

\begin{lemma}[Lemma 3.1 of \cite{Bambusi2020}]\label{lem.g.punto}
	Let $G(t)$ be a smooth in time family of self-adjoint operators and $H(t)$ a continuous in time family of self-adjoint operators. Let $\widetilde U(t) := e^{-\ii G(t)} U_H(t)$. Then $\widetilde U(t) = U_{H'}(t)$, with
	$$
	H'(t) = e^{-\ii G(t)} H(t) e^{\ii G(t)} + \int_0^1 e^{-\ii G(t) s} \partial_t G(t) e^{\ii G(t) s} \textrm{d} s\,.
	$$
\end{lemma}
\begin{lemma}[Iterative lemma]\label{lem.iter}
	Let $\sigma$ as in \eqref{lolly.molly.dolly} and $b, \epsilon, p, \tau, n$ as in the assumptions of Proposition \ref{normal.form}. Let furthermore
	\begin{equation}\label{cento.passi}
 	K := \ln \left( \frac{2e}{\sigma^{p}}\right) \sigma^{-1}\,, \quad n_* := 	K \sigma  -\ln(2e \lambda^b)
    = \ln\left(\lambda^{-b}\sigma^{-p}\right) = \ln \left( \lambda^{-b+\frac{p(1 - b - \epsilon)}{\tau}}\right)\,,
	\end{equation}
	and for any $n=0, \dots, n_*-1$ let
	\begin{equation}\label{regolar_n}
	\kappa_0 := \kappa\,, \quad \kappa_n := \kappa_0 - n \delta\,, \quad \delta := \frac{\kappa_0}{ 2n_* }\,, \quad \hat{C} := \max\{\tC_1, \tC_2\}\,.
	\end{equation}
	There exist two intensive constants $\lambda_0>0$ and $\hat{C}>0$, such that, if $\lambda> \lambda_0$, the following holds.
	If
	\begin{equation}\label{Hn}
	H^{(n)}(\nu \tit) := \lambda^{-1} \left( H_0 + Z_\lambda^{(n)} + V_\lambda^{(n)}(\nu \tit) + R_\lambda^{(n)}(\nu \tit)\right)\,,
	\end{equation}
	with
	\begin{equation}\label{fatta.cosi}
	\begin{gathered}
	\|Z_{\lambda}^{(n)}\|_{\kappa_n} \leq \hat{C} \lambda^{-b} \Big(\sum_{j=0}^{n-1} e^{-j}\Big) \|V\|_{\kappa, C^p} \quad \text{if } n \geq 1\,, \quad  \|Z_{\lambda}^{(0)}\|_{\kappa_0} = 0\,,\\
	\|V_{\lambda}^{(n)}\|_{\kappa_n, \sigma} \leq \hat{C} \lambda^{-b} e^{-n} \|V\|_{\kappa, C^p} \ \ \text{if }\ n \geq 1\,, \quad \|V^{(0)}_\lambda\|_{\kappa_0, \sigma} \leq \hat{C}\|V\|_{\kappa, C^p}\,, \quad \langle V^{(0)}_\lambda \rangle = 0\,, \\
    \|R_{\lambda}^{(n)}\|_{\kappa_n, C^0} \leq \hat{C} \Big(  \sum_{j=0}^{n} e^{-j}\Big) \sigma^{p}\|V\|_{\kappa, C^p}\,, 
	\end{gathered}
	\end{equation}
	then there exists $G^{(n)} \in \cO_{\kappa_{n}, \sigma}$ such that $Y^{(n)}(\nu \tit) := e^{\ii G^{(n)}(\nu \tit)}$ conjugates the dynamics of $H^{(n)}(\nu \tit)$ to the dynamics of
	\begin{equation}
	H^{(n+1)}(\nu \tit) := \lambda^{-1} \left( H_0 + V_\lambda^{(n+1)}(\nu \tit) + R_\lambda^{(n+1)}(\nu \tit)\right)\,,
	\end{equation}
	with
	\begin{equation}\label{cosa.rimane}
	\begin{gathered}
	\|Z_{\lambda}^{(n+1)}\|_{\kappa_{n+1}} \leq \hat{C} \lambda^{-b} \Big(\sum_{j=0}^{n} e^{-j}\Big) \|V\|_{\kappa, C^p}\,, \\	\|V_{\lambda}^{(n+1)}\|_{\kappa_{n+1}, \sigma} \leq \hat{C} \lambda^{-b} e^{-(n+1)}\|V\|_{\kappa, C^p}\,, \quad \|R_{\lambda}^{(n+1)}\|_{\kappa_{n+1}, C^0} \leq \hat{C} \Big(\sum_{j=0}^{n+1} e^{-j} \Big) \sigma^p \|V\|_{\kappa, C^p}\,,
	\end{gathered}
	\end{equation}
	and
	\begin{equation}\label{cosa.cambia}
	\begin{gathered}
	\| Y^{(n)} A\big(Y^{(n)}\big)^*  - A \|_{\kappa_{n+1}, \sigma} \leq  \lambda^{-b} e^{-n} \|V\|_{\kappa, C^p} \|A\|_{\kappa_n, \sigma} \quad \forall A \in \cO_{\kappa, \sigma}\,,\\
	\|  Y^{(n)} B\big(Y^{(n)}\big)^* - B \|_{\kappa_{n+1}, C^0} \leq  \lambda^{-b}  e^{-n} \|V\|_{\kappa, C^p} \|B\|_{\kappa_n, C^0} \quad \forall B \in \cO_{\kappa, C^0}\,.
	\end{gathered}
	\end{equation}
\end{lemma}
 \begin{remark}
 Parameters in \eqref{cento.passi} are chosen in such a way that
 $$
  \frac{K^{\tau} n_*}{\lambda} < \frac{1}{\lambda^b}\,, \quad  e^{-n_*} \lambda^{-b} \leq \sigma^{p} \,, \quad
    e^{-K \sigma} \leq \frac{1}{2e} e^{-n_*} \lambda^{-b} \,.
 $$
 The first condition ensures $\delta^{-1}\|G^{(0)}\|_{\kappa_0, \sigma} \lesssim \lambda^{-b} \|V\|_{\kappa, C^p}$, namely that the transformations $Y^{(n)}$ are close to the identity; the second and third conditions ensure $\|V^{(n_*)}_\lambda\|_{\kappa_{n_*}, \sigma} \leq \|R_\lambda\|_{\kappa_{0}, C^0}$ and $  \|(V^{(0)}_\lambda)^{\uv}\|_{\kappa_0, \sigma} \leq \frac{1}{2e } \|V^{(n_*)}_\lambda\|_{\kappa_{n_*}, \sigma} $\,, namely that the contributions coming from the analytic smoothing are dominant in size with respect to the remainders produced by the commutator expansion \eqref{eq:commuta.che.ti.passa} and the tails of the ultraviolet cut-off \eqref{eq:ir.uv}.
 \end{remark}
\begin{proof}
	Suppose $H^{(n)}$ is as in \eqref{Hn}-- \eqref{fatta.cosi}.
	Let  $G^{(n)}$ be the solution of equation
	\begin{equation}\label{homo}
	(\nu \cdot \partial_\varphi) G^{(n)}  + \lambda^{-1}V^{(n)}_\lambda = \lambda^{-1} \big(V^{(n)}_\lambda\big)^\uv + \lambda^{-1} \big\langle V_\lambda^{(n)} \big\rangle\,.
	\end{equation}
	From \eqref{fatta.cosi}, it follows that $\lambda^{-1} V_\lambda^{(n)} \in \cO_{\kappa_n, \sigma}$. Applying  Lemma \ref{lemma.homo}, $G^{(n)} \in \cO_{\kappa_n, \sigma}$ is well-defined and one has
    \begin{equation}\label{che.bella.g}
	\|G^{(n)}\|_{\kappa_n, \sigma} \leq \frac{ K^{\tau}}{\gamma \lambda} \|V^{(n)}_\lambda\|_{\kappa_n, \sigma}\,.
	\end{equation}
	Therefore, taking $\delta$, $n_*$ and $K$ as in \eqref{regolar_n}, \eqref{cento.passi}, and $\sigma$ as in \eqref{lolly.molly.dolly}, there exists $\lambda_1 = \lambda_1(b, \epsilon, \tau, p, \gamma, \kappa_0)$ such that, if $\lambda > \lambda_1$,
    $$
    \begin{aligned}
	\frac{4 e^{-\kappa_{n+1}} \|G^{(n)}\|_{\kappa_n, \sigma}}{\delta} &\overset{\text{\eqref{regolar_n}}}{=} \frac{8 e^{-\kappa_{n+1}}  n_* \|G^{(n)}\|_{\kappa_n, \sigma}}{\kappa_0}
   \overset{\text{\eqref{che.bella.g}, \eqref{cento.passi}}}{\leq} \frac{8 e^{-\kappa_{n+1}}  K^{\tau + 1} \sigma \|V^{(n)}_\lambda\|_{\kappa_n, \sigma}}{\kappa_0 \gamma \lambda}\\
	&\overset{\eqref{cento.passi}}{=} \frac{8 e^{-\kappa_{n+1}} \ln^{\tau+1}(2 e\lambda^{-b} \sigma^{-p}) }{\kappa_0 \gamma \lambda  \sigma^{\tau} } \|V^{(n)}_\lambda\|_{\kappa_n, \sigma}\\
	&\overset{\eqref{lolly.molly.dolly}}{=} \frac{8 e^{-\kappa_{n+1}} }{\kappa_0 \gamma  } \lambda^{-b- \epsilon} \ln^{\tau+1}(2 e \lambda^{-b+\frac{p(1-b - \epsilon)}{\tau}}) \|V^{(n)}_\lambda\|_{\kappa_n, \sigma}\,.
    \end{aligned}
    $$
    Then if $n=0$ we deduce
    \begin{equation}\label{g0.small}
    \begin{aligned}
        \frac{4 e^{-\kappa_{1}} \|G^{(0)}\|_{\kappa_0, \sigma}}{\delta} \stackrel{\eqref{fatta.cosi}}{\leq} \frac{8 e^{-\kappa_{1}} }{\kappa_0 \gamma  } \lambda^{-b- \epsilon} \ln^{\tau+1}(2 e \lambda^{-b+\frac{p(1-b - \epsilon)}{\tau}}) \hat{C} \|V\|_{\kappa, C^p} \leq \lambda^{-b-\frac{\epsilon}{2}} \|V\|_{\kappa, C^p} < \frac{1}{2}\,,
    \end{aligned}
    \end{equation}
    and analogously if $n \geq 1$ 
	\begin{equation}\label{gn.small}
	\begin{aligned}
    \frac{4 e^{-\kappa_{n+1}} \|G^{(n)}\|_{\kappa_n, \sigma}}{\delta}
	& \overset{\eqref{fatta.cosi}}{\leq} \left(\frac{8 \hat{C} e^{-\kappa_{n+1}}}{\gamma \kappa_0 } \lambda^{-b} \ln^{\tau + 1}(2e \lambda^{-b+\frac{p(1-b - \epsilon)}{\tau}}) \right) \lambda^{-b-\epsilon} e^{-n} \|V\|_{\kappa, C^p}\\
	& \leq \lambda^{-b - \frac{\epsilon}{2}} e^{-n} \|V\|_{\kappa, C^p}< \frac 1 2\,.
	\end{aligned}
	\end{equation}
	Furthermore, by Lemma \ref{lem.g.punto}, the unitary $Y^{(n)}:= e^{-\ii G^{(n)}}$ conjugates the dynamics of $H^{(n)}$ to the dynamics of
	$$
	\begin{aligned}
	H^{(n+1)} &= e^{-\ii G^{(n)}} H^{(n)} e^{\ii G^{(n)}} + \int_{0}^1 e^{-\ii sG^{(n)}} (\nu \cdot \partial_\varphi) G^{(n)} e^{\ii sG^{(n)}}\ \ud s\\
	&=:	\lambda^{-1} \left(H_0  + Z_\lambda^{(n+1)} + V_\lambda^{(n+1)} + R_\lambda^{(n+1)}\right)\,,\\
	\end{aligned}
	$$
	where we have set
	\begin{equation}\label{eq:new.z}
	Z_\lambda^{(n+1)} := Z_\lambda^{(n)} + \big \langle V_\lambda^{(n)}\big \rangle\,,
	\end{equation}
	\begin{align}
		 V_\lambda^{(n+1)} &:=  e^{-\ii G^{(n)}} H_0 e^{\ii G^{(n)}} - H_0\\
		 &+  e^{-\ii G^{(n)}} Z_\lambda^{(n)} e^{\ii G^{(n)}} - Z_\lambda^{(n)}\\
		 & + \lambda \int_{0}^1 \left(e^{-\ii sG^{(n)}} (\nu \cdot \partial_\varphi) G^{(n)} e^{\ii sG^{(n)}} - (\nu \cdot\partial_{\varphi})G^{(n)} \right)\ud s \\
		 & + e^{-\ii G^{(n)}} V^{(n)}_\lambda e^{\ii G^{(n)}} -V^{(n)}_\lambda\,,
	\end{align}
	and
	\begin{align}\label{rangers}
	R^{(n+1)}_\lambda := e^{-\ii G^{(n)}} R^{(n)}_\lambda e^{\ii G^{(n)}} +  (V^{(n)}_\lambda)^\uv \,.
	\end{align}
    We start by estimating $Z^{(n+1)}_\lambda$. We treat the cases $n=0$ and $n\geq 1$ separately. For $n=0$, we observe that, since $\langle V \rangle = \langle V_\lambda \rangle = 0$,
    by \eqref{eq:new.z} and using that $\langle V^{(0)}_\lambda \rangle = 0$, one has $Z^{(1)}_\lambda = \langle V^{(0)}_\lambda \rangle = 0$\,. If $n \geq 1$, we observe that, by the inductive hypotheses \eqref{fatta.cosi} on $Z_\lambda^{(n)}$ and on $V_\lambda^{(n)}$ and by Remark \ref{vattelapesca}, one immediately has
	$$
	\|Z^{(n+1)}_\lambda\|_{\kappa_{n+1}} \leq \|Z_\lambda^{(n)}\|_{\kappa_n} + \| \langle V_\lambda^{(n)} \rangle \|_{\kappa_n, \sigma} \leq \hat C \lambda^{-b}\sum_{j=0}^{n-1} e^{-j} \|V\|_{\kappa, C^p} + \hat C \lambda^{-b} e^{-n} \|V\|_{\kappa, C^p}= \hat C \sum_{j=0}^{n} e^{-j} \|V\|_{\kappa, C^p} \,,
	$$
	which gives the desired estimate on $Z_\lambda^{(n+1)}$. We now prove that also $V_\lambda^{(n+1)}, R_\lambda^{(n+1)}$ satisfy the smallness assumptions in \eqref{cosa.rimane}.\\
	By Item (i) of Lemma \ref{lemma.commut} with $\eta=\frac{1}{2}$, $\kappa= \kappa_{n+1}$ and up to enlarging again $\lambda$, we have
	\begin{equation}\label{chuck}
	\begin{aligned}
	\|e^{-\ii G^{(n)}} H_0 e^{\ii G^{(n)}} - H_0\|_{\kappa_{n+1}, \sigma} &\leq \frac{8 e^{-\kappa_{n+1}}}{\delta} \|G^{(n)}\|_{\kappa_n, \sigma} \|H_0\|_{\kappa_n}\\
	& \hspace{-0.2cm}\overset{\text{\eqref{gn.small}}}{\leq}  2 \lambda^{-b - \frac{\epsilon}{2}} e^{-n} \|V\|_{\kappa, C^p} \|H_0\|_{\kappa_0} \leq \frac {\hat{C}} {4} \lambda^{-b} e^{-(n+1)}  \|V\|_{\kappa, C^p}\,.
	\end{aligned}
	\end{equation}
	Recalling the inductive hypothesis \eqref{fatta.cosi} on $Z_\lambda^{(n)}$, if $n \geq 1$ we also have
	\begin{equation}\label{norris}
	\begin{aligned}
	\| e^{-\ii G^{(n)}} Z_\lambda^{(n)} e^{\ii G^{(n)}} - Z_\lambda^{(n)}\|_{\kappa_{n+1}, \sigma} &\leq \frac{8 e^{-\kappa_{n+1}}}{\delta} \|G^{(n)}\|_{\kappa_n, \sigma} \|Z_\lambda^{(n)}\|_{\kappa_n}\\
	& \stackrel{\eqref{gn.small}}{\leq}  2 \lambda^{-b-\frac{\epsilon}{2}} e^{-n} \Vert V \Vert_{\kappa,C^p} \hat{C} \frac{e}{e-1} \lambda^{-b} \Vert V \Vert_{\kappa,C^p} \\
    &\leq \frac {\hat{C}} {4} \lambda^{-b} e^{-(n+1)}  \|V\|_{\kappa, C^p}\,. 
	\end{aligned}
	\end{equation}
    Note that, if $n=0$, the term estimated in \eqref{norris} vanishes, since $Z^{(0)}_\lambda = 0$.
	Analogously, using the inductive hypothesis \eqref{fatta.cosi} on $V^{(n)}$ and recalling that, by definition of $G^{(n)}$ and by Lemma \ref{lemma.uv}, $\|(\nu \cdot \partial_\varphi) G^{(n)}\|_{\kappa_n, \sigma} = \lambda^{-1} \|\big(V^{(n)}_\lambda\big)^{\ir} - \langle V_\lambda^{(n)} \rangle \|_{\kappa_n, \sigma} \leq \lambda^{-1} \|V_\lambda^{(n)}\|_{\kappa_n, \sigma}\,,$ for any $n \geq 0$ we have
	\begin{equation}\label{james}
	\begin{aligned}
	\lambda \left \| \int_{0}^1 \left(e^{-\ii sG^{(n)}} (\nu \cdot \partial_\varphi) G^{(n)} e^{\ii sG^{(n)}} - (\nu \cdot\partial_{\varphi})G^{(n)} \right)\ud s\right\|_{\kappa_{n+1}, \sigma} &\leq \lambda \frac{8 e^{-\kappa_{n+1}}}{\delta} \|G^{(n)}\|_{\kappa_n, \sigma} \|(\nu \cdot \partial_\varphi) G^{(n)}\|_{\kappa_n, \sigma} \\
	&\leq 2 \lambda^{1-b - \frac{\epsilon}{2}} e^{-n} \|V\|_{\kappa, C^p} \lambda^{-1} \|V^{(n)}_\lambda\|_{\kappa_n, \sigma} \\
	&\stackrel{\eqref{fatta.cosi}}{\leq} \left(2 \lambda^{-b - \frac{\epsilon}{2}} e^{-n}\|V\|_{\kappa, C^p} \right) \hat{C} e^{-n} \|V\|_{\kappa, C^p} \\
	&\leq \frac{\hat{C}}{ 4} \lambda^{-b} e^{-(n+1)} \|V\|_{\kappa, C^p}\,,
	\end{aligned}
	\end{equation}
	provided $\lambda$ is large enough.
	With the same arguments for any $n \geq 0$ we prove
	\begin{equation}\label{trivette}
	\begin{aligned}
	\| e^{-\ii G^{(n)}} V_\lambda^{(n)} e^{\ii G^{(n)}} - V_\lambda^{(n)} \|_{\kappa_{n+1}, \sigma} &\leq \frac{8 e^{-\kappa_{n+1}}}{\delta} \|G^{(n)}\|_{\kappa_n, \sigma} \|V^{(n)}_\lambda\|_{\kappa_n, \sigma} \\
	&\leq \left(2 \lambda^{-b -\frac{\epsilon}{2}} e^{-n} \|V\|_{\kappa, C^p}\right) \hat{C} e^{-n} \|V\|_{\kappa, C^p} \\
    & \leq  \frac {\hat{C}}{ 4} \lambda^{-b} e^{-(n+1)} \|V\|_{\kappa, C^p}\,,
	\end{aligned}
	\end{equation}
	provided $\lambda$ is large enough. Combining estimates \eqref{chuck}, \eqref{norris}, \eqref{james}, \eqref{trivette}, we get
	the estimate on $V^{(n+1)}$ in \eqref{cosa.rimane}. 
    To prove the last inequality in \eqref{cosa.rimane}, we use Item (ii) in Lemma \ref{lemma.commut}, together with the smallness assumption \eqref{gn.small} previously proved for $G^{(n)}$. This allows us to conclude that, for any $n \geq 0 $,
	\begin{equation}\label{texas}
	\begin{aligned}
	\| e^{-\ii G^{(n)}} R^{(n)}_\lambda e^{\ii G^{(n)}} - R^{(n)}_\lambda \|_{\kappa_{n+1}, C^0} & \leq \frac{8 e^{-\kappa_{n+1}}}{\delta} \|G^{(n)}\|_{\kappa_n, C^0} \|R^{(n)}_\lambda\|_{\kappa_n, C^0}\\
	&\leq \frac{8 e^{-\kappa_{n+1}}}{\delta} \|G^{(n)}\|_{\kappa_n, C^0} \hat{C} \big(\sum_{j=0}^{n} e^{-j} \big) \sigma^p \|V\|_{\kappa, C^p} \\
	&\leq \left(2 \lambda^{-b -  \frac{\epsilon}{2}} \|V\|_{\kappa, C^p} \big(\sum_{j=0}^{n} e^{-j} \big) \right) \hat{C} e^{-n} \sigma^p \|V\|_{\kappa, C^p}\\
	&\leq \frac {\hat{C}}{ 2} e^{-(n+1)} \sigma^p \|V\|_{\kappa, C^p}\,.
	\end{aligned}
	\end{equation}
	Furthermore, by Lemma \ref{lemma.uv} and recalling the definition of $n_*$ in \eqref{cento.passi} , one has
	\begin{equation}\label{rangers.1}
	 \| \big(V_\lambda^{(n)}\big)^\uv\|_{\kappa_{n+1}, C^0} \leq \| \big(V_\lambda^{(n)}\big)^\uv\|_{\kappa_{n}, 0}  \leq e^{-K \sigma} \|V_\lambda^{(n)}\|_{\kappa_{n}, \sigma} = \frac{\sigma^{p}}{2 e} \|V_\lambda^{(n)}\|_{\kappa_{n}, \sigma}\leq \frac{\sigma^p}{2 e } \hat{C} e^{-n} \|V\|_{\kappa, C^p}\,.
	\end{equation}
	Then, combining \eqref{texas}, the inductive estimate in \eqref{fatta.cosi} and equation \eqref{homo}, and using Lemma \ref{lemma.uv}, we get that $R^{(n+1)}_\lambda$ defined in \eqref{rangers} satisfies
	\begin{equation}
	\begin{aligned}
	\| R^{(n+1)}_\lambda\|_{\kappa_{n+1}, C^0} &\leq \| e^{-\ii G^{(n)}} R^{(n)}_\lambda e^{\ii G^{(n)}} - R^{(n)}_\lambda\|_{\kappa_{n+1}, C^0} + \| R^{(n)}_\lambda\|_{\kappa_{n+1}, C^0} + \|(V_\lambda^{(n)})^\uv\|_{\kappa_{n+1}, C^0}\\
	&\leq \frac{\hat{C}}{ 2} e^{-(n+1)}  \sigma^p \|V\|_{\kappa, C^p} + \Big(\sum_{j=0}^{n} e^{-j} \Big) \sigma^p \|V\|_{\kappa, C^p} +  \frac{\sigma^p}{2 e } \hat{C} e^{-n} \|V\|_{\kappa, C^p}\\
	&\leq \Big(\sum_{j=0}^{n+1} e^{-j} \Big) \hat{C} \sigma^p \|V\|_{\kappa, C^p}\,,
	\end{aligned}
	\end{equation}
	which proves the second estimate in \eqref{cosa.rimane}. Finally, estimate \eqref{cosa.cambia} follows from the smallness condition \eqref{gn.small} on $G^{(n)}$ and Items (i)--(ii) of Lemma \ref{lemma.commut}.
 \end{proof}
 \begin{proof}[Proof of Proposition \ref{normal.form}]
 	Let $H^{(0)} (\nu t) := \lambda^{-1} H(\nu t)$ and recall the definition of $\kappa_0$ in \eqref{regolar_n}. By Corollary \ref{cor.lissage}, one has
 	$$
 	H^{(0)} = \lambda^{-1} \left( H_0 + V_\lambda^{(0)} + R_\lambda^{(0)}\right)\,, \quad \|V^{(0)}_\lambda\|_{\kappa_0, \sigma} \leq \tC_1 \|V\|_{\kappa, C^p}\,, \quad \langle V^{(0)}_\lambda \rangle = 0\,, \quad \|R^{(0)}_\lambda\|_{\kappa_0, C^0} \leq \tC_2 \sigma^p \|V\|_{\kappa, C^p}\,.
 	$$
 	Applying iteratively Lemma \ref{lem.iter} $n_*$ times, with $n_*$ defined in \eqref{cento.passi}, we get that the map
 	$$
 	Y^{(\fin)} := Y^{(n_*-1)} \circ \cdots \circ Y^{(0)}
 	$$ conjugates the dynamics of $\lambda^{-1} H \equiv H^{(0)}$ to the dynamics of $H^{(\fin)} := H^{(n_*)}$.
 	Then, using estimate \eqref{cosa.rimane} and recalling $\kappa_\fin = \kappa_{n_*}$, one has
 	$$
 	\begin{gathered}
 	\|Z_\lambda^{(n_*)}\|_{\kappa_{\fin}} \leq \hat C \lambda^{-b} \sum_{j=0}^{n_*-1} e^{-j} \|V\|_{\kappa, C^p} \leq \frac{\hat C e}{e-1} \lambda^{-b} \|V\|_{\kappa, C^p}\,,\\
 	\|V^{(n_*)}_\lambda\|_{\kappa_{\fin}, \sigma} \leq \hat C \lambda^{-b} e^{-n_*} \|V\|_{\kappa, C^p} < \hat C \sigma^p \|V\|_{\kappa, C^p}\,, \quad \|R^{(n_*)}_\lambda\|_{\kappa_{\fin}, C^0} \leq \frac{\hat C e }{e-1} \sigma^p \|V\|_{\kappa, C^p} \,,
 	\end{gathered}
 	$$
 	which gives Item 1. To prove Item 2, we define
 	$$
 	{Y}^{(\leq n)}:=  Y^{(n)}\circ \cdots \circ  Y^{(0)}\,,
 	$$
 	and we prove inductively on $n$ that
 	\begin{equation}\label{comp.yn}
 	\begin{split}
 	\| Y^{(\leq n)} A \big(Y^{(\leq n)}\big)^* - A \|_{\kappa_{n+1}, \sigma} &\leq 2 \lambda^{-b} \sum_{j=0}^{n} e^{-j} \|V\|_{\kappa, C^p} \|A\|_{\kappa_0, \sigma} \quad \forall A \in \cO_{\kappa, \sigma}\,, \\
 	\| Y^{(\leq n)} B \big(Y^{(\leq n)}\big)^* - B \|_{\kappa_{n+1}, C^0} &\leq 2  \lambda^{-b} \sum_{j=0}^{n} e^{-j} \|V\|_{\kappa, C^p} \|B\|_{\kappa_0, C^0} \quad \forall B \in \cO_{\kappa, C^0}\,.
 	\end{split}
 	\end{equation}
 	Then Item 2 follows since $Y^{(\fin)}= Y^{(\leq n_*-1)}$ and noting that $\sum_{j=0}^{n_*-1} e^{-j} \leq \sum_{j=0}^{\infty} e^{-j} = \frac{e}{e-1}$.
   To prove \eqref{comp.yn}, we note that when $n=0$, we have that $Y^{(\leq 0)} = Y^{(0)}$, and the estimate follows by taking $n=0$ in \eqref{cosa.cambia}. If \eqref{comp.yn} is true for $Y^{(\leq n -1)}$, then $Y^{(\leq n)} = Y^{(\leq n-1)} \circ Y^{(n)}  $ satisfies
 	\begin{align}
 	\nonumber
  	Y^{(\leq n)}A &	\big(Y^{(\leq n)}\big)^*  - A =  Y^{(n)} \left[Y^{(\leq n-1)} A \big(Y^{(\leq n-1)}\big)^*  - A \right] \big(Y^{(n)}\big)^*  +  Y^{(n)}A \big(Y^{(n)}\big)^*  - A\\
 	\label{buona}
 	&=  Y^{(n)} \left[Y^{(\leq n-1)} A \big(Y^{(\leq n-1)}\big)^*  - A \right] \big(Y^{(n)}\big)^* - \left[ Y^{(\leq n-1)} A \big(Y^{(\leq n-1)}\big)^*  - A\right]\\
 	\label{notte}
 	& \quad + \left[ Y^{(\leq n-1)}  A \big(Y^{(\leq n-1)}\big)^* - A\right] \\
 	\label{fiorellino}
 	&\quad  +   Y^{(n)} A\big(Y^{(n)}\big)^* - A\,,
 	\end{align}
 	where:
 	\begin{align*}
 	\left\| \eqref{buona}\right\|_{\kappa_{n+1}, \sigma} &\leq \lambda^{-b} e^{-n} \|V\|_{\kappa, C^p} \left\| Y^{(\leq n-1)} A \big(Y^{(\leq n-1)}\big)^*  - A \right\|_{\kappa_n, \sigma}\\
 	& \leq  \left( 2 \lambda^{-b} \sum_{j=0}^{n-1} e^{-j} \|V\|_{\kappa, C^p} \right)\lambda^{-b} e^{-n} \|V\|_{\kappa, C^p} \|A\|_{\kappa_0} \leq \lambda^{-b} e^{-n} \|V\|_{\kappa, C^p} \|A\|_{\kappa_0}
 	\end{align*}
 	due to estimate \eqref{cosa.cambia} and the inductive hypothesis,
 	\begin{align*}
 	\|\eqref{notte}\|_{\kappa_{n+1} ,\sigma} &\leq 2 \lambda^{-b} \sum_{j=0}^{n-1} e^{-j} \|V\|_{\kappa, C^p}  \|A\|_{\kappa_0, C^p}
 	\end{align*}
 	due to the inductive hypothesis, and
 	\begin{align*}
 	\|\eqref{fiorellino}\|_{\kappa_{n+1}, \sigma} \leq  \lambda^{-b} e^{-n} \|V\|_{\kappa, C^p} \|A\|_{\kappa_0, C^p}\,
 	\end{align*}
	again by estimate \eqref{cosa.cambia}. Then one has
	$$
	\left\| Y^{(\leq n)} A  \big(Y^{(\leq n)}\big)^* - A \right\|_{\kappa_{n+1}, \sigma} \leq \| \eqref{buona}\|_{\kappa_{n+1}, \sigma} +  \| \eqref{notte}\|_{\kappa_{n+1}, \sigma} +  \| \eqref{fiorellino}\|_{\kappa_{n+1}, \sigma}\,,
	$$
	which proves \eqref{comp.yn} in the analytic case. The estimate for $C^0$ operators follows in the same way.
 \end{proof}

\subsection{Slow heating}
In the following, we will use the notation
    \begin{equation}\label{eq:heff}
    H_{\eff}:= H_0 + Z_\lambda^{(\fin)}\,,
    \end{equation}
where $Z_\lambda^{(\fin)}$ is defined in Proposition \ref{normal.form}. Note that, by Item 1 of the same proposition, there exists an intensive constant $\hat C> 0$ such that 
\begin{equation}\label{eq:Heff.est}
    \|H_\eff - H_0\|_{\kappa_\fin}\leq  \hat{C} \lambda^{-b} \|V\|_{\kappa, C^p}\,.
\end{equation}
From Proposition \ref{normal.form} we deduce:
\begin{lemma}
	There exists an intensive constant $D_1>0$ such that 
    \begin{equation}\label{h0.h0.h0}
	|\Lambda|^{-1}\| U_H^*(t) H_0 U_H(t) - H_{0} \|_{\operatorname{op}} \leq D_1 \lambda^{-b} \quad \forall 0< t < \lambda^{-b+\frac{p(1-b-\epsilon)}{\tau}}\,.
	\end{equation}
\end{lemma}
\begin{proof}
    First we observe that
        \[
             U^*_{H}(t) H_0 U_H(t)-H_0 =U^*_{H}(t) (H_0-H_\eff) U_H(t) +U_H^*(t) H_\eff U_H(t)-H_\eff -H_0+H_\eff \, .
        \]
    Then, by triangular inequality and by unitarity of $U_H$, we have that
        \begin{equation}
            \label{eq:salsicce.di.lucio}
        \Vert U_H^*(t) H_0 U_H(t)-H_0 \Vert_{\operatorname{op}} \leq 2 \Vert H_0-H_{\eff} \Vert_{\operatorname{op}}+\Vert U_H^*(t) H_\eff U_H(t) -H_\eff \Vert_{\operatorname{op}} \, .
        \end{equation}
        The first term of \eqref{eq:salsicce.di.lucio} is estimated using Remark \ref{rmk:frecciarotta} and \eqref{eq:Heff.est}, indeed
        \[
            \Vert H_0-H_\eff \Vert_{\operatorname{op}} \leq |\Lambda| \Vert H_0-H_\eff \Vert_{0} \leq |\Lambda| \hat{C} \lambda^{-b} \Vert V \Vert_{\kappa,C^p} \, .
        \]
        To estimate the second one, we first recall the notation $U_H(t)=\tU_H(\tit)$ with $t=\lambda^{-1} \tit$. Moreover, as a consequence of Proposition \ref{normal.form}, for any $\tit \in \mathbb{R}$ we have
	\begin{equation}\label{eq:Disuguamagica}
	\begin{aligned}
	\Big\|\tU^*_{H^{(\fin)}}&(\tit) H_{\eff} \tU_{H^{(\fin)}}(\tit) -H_{\eff} \Big\|_{0, C^0} \\
    &= \left\|\int_0^\tit	\tU^*_{H^{(\fin)}}(s) \lambda^{-1} \left[ H_{\eff},  H_{\eff} +V_\lambda^{(\fin)}(\nu s) + R_\lambda^{(\fin)}(\nu s)\right] \tU_{H^{(\fin)}}(s) 	\ud s \right\|_{0, C^0} \\
	&\leq 2 |\tit| \Vert H_{\eff} \Vert_{0} \lambda^{-1} \left(\Vert V_\lambda^{(\fin)} \Vert_{0,C^0} + \Vert R_\lambda^{(\fin)} \Vert_{0,C^0} \right)\\
	&\hspace{-1.15cm}\overset{\text{Item 2, Proposition \ref{normal.form}}}{\leq} 2 \, |\tit| \, D  \, \Vert H_{\eff} \Vert_0 \, \lambda^{-1} \lambda^{-\frac{p(1-b-\epsilon)}{\tau}} \Vert V \Vert_{\kappa,C^p} \, ,
	\end{aligned}
	\end{equation}
    where $H^{(\fin)}$ is defined in \eqref{eq:Hfin}. Furthermore, by \eqref{def.conj}, one has
	$$
	\begin{aligned}
	\tU^*_{H}(\tit)  H_{\eff} \tU_{H}(\tit) - H_{\eff} &= \tU^*_{H^{(\fin)}}(\tit) Y^{(\fin )}(\nu \tit)  H_{\eff} (Y^{(\fin )}(\nu \tit))^* \tU_{H^{(\fin)}}(\tit) -  H_{\eff} \\
	&= \tU^*_{H^{(\fin)}}(\tit) \left(Y^{(\fin )}(\nu \tit)  H_{\eff} (Y^{(\fin )}(\nu \tit))^* - H_{\eff}\right) \tU_{H^{(\fin)}}(\tit)\\
	& + \tU^*_{H^{(\fin)}}(\tit) H_{\eff} \tU_{H^{(\fin)}}(\tit) -  H_{\eff}\,.
	\end{aligned}
	$$
    We first use the latter equation for a triangular inequality
	\[
		\begin{aligned}
	|\Lambda|^{-1} \left\|\tU^*_{H}(\tit)  H_{\eff} \tU_{H}(\tit) -  H_{\eff}\right\|_{\operatorname{op}}&\leq |\Lambda|^{-1}\Vert Y^{(\fin)}(\nu \tit) H_{\eff} (Y^{(\fin)}(\nu \tit))^*-H_{\eff} \Vert_{\operatorname{op}} \\
	&\qquad + |\Lambda|^{-1} \Vert U_{H^{(\fin)}}^*(\tit) H_{\eff} U_{H^{(\fin)}}(\tit)-H_{\eff} \Vert_{\operatorname{op}} \\
	& \leq \frac{2e}{e-1}\lambda^{-b} \Vert H_{\eff} \Vert_{\kappa_{\fin}} + 2 |\tit| D \Vert H_{\eff} \Vert_{\kappa_{\fin}} \lambda^{-1} \lambda^{-\frac{p(1-b-\epsilon)}{\tau}} \Vert V \Vert_{\kappa,C^p} \\
	& \leq \left(\frac{2e}{e-1} \lambda^{-b} + 2 |t| D \lambda^{-\frac{p(1-b-\epsilon)}{\tau}}\Vert V \Vert_{\kappa,C^p}\right) \Vert H_{\eff} \Vert_{\kappa_{\fin}}\, ,
		\end{aligned}
	\]
	where we used Remark \ref{rmk:frecciarotta}, \eqref{eq:Disuguamagica}, Item 3 in Proposition \ref{normal.form} and the relation $\tit=\lambda t$. Finally we observe that, by \eqref{eq:Heff.est}, provided $\lambda$ is large enough one has
    $\|H_\eff\|_{\kappa_\fin} \leq \|H_0\|_{\kappa_\fin} + \|H_{\eff} - H_0\|_{\kappa_\fin} \leq 2 \|H_0\|_{\kappa}$. Then for times $|t| \leq \lambda^{-b+\frac{p(1-b-\epsilon)}{\tau}}$ we have proved \eqref{h0.h0.h0}, with a constant $D_1$ given by
	\[
		D_1= 2 \hat{C} \Vert V \Vert_{\kappa,C^p}+ 2\left( \frac{2e}{e-1} + 2 D \Vert V \Vert_{\kappa,C^p} \right) \Vert H_{0} \Vert_{\kappa} \, .
	\]
\end{proof}

\subsection{Dynamics of local observables}
In this Subsection we prove Item $(ii)$ of Theorem \ref{thm:main}.
We first state the following result, which is an improvement of Lemma 6.2 in \cite{Gallone-Langella-2024} (see also  \cite{Abanin2017} ) and it is essentially based on Lieb-Robinson bounds:
\begin{lemma}\label{lem:LemmaSempreDiverso}
    Let $O$ be a local observable acting within $S_O \in \mathcal{P}_c(\Lambda)$, $A \in \mathcal{O}_{\kappa, C^0},$ and $Z \in \mathcal{O}_{2 \kappa}$. Then there exists a positive constant $C=C(|S_O|, d, \kappa)$ such that
    \begin{equation}
        \label{eq:int.est}
        \int_0^t \ud s \left\| [A(\lambda \nu s)\,,\ e^{\ii s Z} O e^{-\ii s Z} ]\right\|_{\operatorname{op}} \leq C \langle \Vert Z \Vert_{2 \kappa} \rangle^{d}\langle t \rangle^{d+1} \|O\|_{\operatorname{op}} \|A\|_{\kappa, C^0}\,.
    \end{equation}
\end{lemma}
\begin{proof}
        This proof is the same of \cite[Lemma 6.2]{Gallone-Langella-2024}, apart for a couple of differences: the improvement in the exponent in time and the use of the norm $\Vert \cdot \Vert_{\kappa,C^0}$ instead of $\Vert \cdot \Vert_{\kappa,\sigma}$.

	First, we use $A(\lambda \nu t)=\sum_{S \in \mathcal{P}_c(\Lambda)} A_S(\lambda \nu t)$ and triangular inequality to have
	\[
	\Vert [A(\lambda \nu t), e^{-\ii s Z} O e^{\ii s Z}] \Vert_{\mathrm{op}} \leq \sum_{x \in \Lambda} \sum_{\substack{S \in \mathcal{P}_c(\Lambda)\\ {x \in S}}} \Vert [A_S(\lambda \nu t), e^{-\ii s Z} O e^{\ii s Z}] \Vert_{\mathrm{op}} =: (\star) \, .
	\]
	We now define $Q_{S_O}$ as the smallest ball that contains $S_O$; we call $r_Q$ its radius. Then we construct $B_{S_O}$ which is the ball of radius $r_{Q}+vs$ and the same center as $Q_{S_O}$. Then by Lieb-Robinson bounds \cite[Lemma 6.1]{Gallone-Langella-2024} we get for any fixed $s>0$
	\[
	(\star) \leq 2\sum_{x \in  B_{S_O}} \sum_{\substack{S \in \mathcal{P}_c(\Lambda) \\  x \in S}} \Vert A_S(\lambda \nu s) \Vert_{\mathrm{op}} \Vert O \Vert_{\mathrm{op}} + \sum_{x \in \Lambda \setminus B_{S_O}} \sum_{\substack{S \in \mathcal{P}_c(\Lambda) \\  x \in S}} \Vert A_S(\lambda \nu s) \Vert_{\mathrm{op}} \Vert O \Vert_{\mathrm{op}} |S_O| e^{- \kappa (d(S,S_O)-v s))} = :(*)\,,
	\]
    with $v:= C_0(d) \kappa^{-(d+2)} e^\kappa \|Z\|_{2\kappa}$.
	 For any $x \in S$:
	\[
	\begin{split}
	d(S,S_O)-vs & \geq d(x,S_O)-vs-|S| \geq d(x, Q_{S_O}) - v s - |S| = d(x,B_{S_O})-|S|
	\end{split}
	\]
	which means
	\[
	e^{-\kappa(d(S,S_O)-vs)} \leq e^{-\kappa d(x,B_{S_O})} e^{\kappa |S|}
	\]
	and then
	\begin{equation}\label{eq:salsa.verde}
	    (*) \leq 2\sum_{x \in  B_{S_O}} \sum_{\substack{S \in \mathcal{P}_c(\Lambda) \\  x \in S}} \Vert A_S( \lambda \nu s) \Vert_{\mathrm{op}} \Vert O \Vert_{\mathrm{op}}+ \sum_{x \in \Lambda \setminus B_{S_O}} \sum_{\substack{S \in \mathcal{P}_c(\Lambda) \\  x \in S}} \Vert A_S (\lambda \nu s) \Vert_{\mathrm{op}} \Vert O \Vert_{\mathrm{op}} |S_O| e^{\kappa |S|} e^{-\kappa d(x,B_{S_O})} \, .
	\end{equation}
	For any $\ell \in \mathbb{N}$, we define $C_\ell  = \{x \in \Lambda \,  | \, \ell < d(x,B_{S_O}) \leq \ell+1 \}$ and we note that there exists $C(d)>0$ such that one can estimate
    \[
        |B_{S_O}| \leq C(d) (|S_O|+v s|)^d \, , \qquad |C_\ell| \leq  C(d)
        (|S_O|+v s + \ell +1)^{d} \, .
    \]
    Then the sum over $x \in B_{S_O}$ appearing at the right hand side in \eqref{eq:salsa.verde} can be bounded by
    \begin{equation}\label{eq:mondeghili}
        2 C(d){(|S_O|+v s)^d}\Vert O \Vert_{\mathrm{op}} \Vert A(\lambda \nu s) \Vert_{0}\,.
    \end{equation}
    We now estimate the other sum. One has
	\[
	\begin{split}
	\sum_{x \in \Lambda \setminus B_{S_O}} \sum_{\substack{S \in \mathcal{P}_c(\Lambda) \\  x \in S}} &\Vert A_S(\lambda \nu s) \Vert_{\mathrm{op}} \Vert O \Vert_{\mathrm{op}} |S_O| e^{\kappa |S|} e^{-\kappa d(x,B_{S_O})}  =  \\
    &\leq \sum_{\ell \geq 0} \sum_{x \in C_\ell} \sum_{\substack{ S \in \mathcal{P}_c(\Lambda) \\  x \in S}} \Vert A_S(\lambda \nu s) \Vert_{\mathrm{op}} \Vert O \Vert_{\mathrm{op}} |S_O| e^{\kappa|S|} e^{-\kappa\ell} \\
	&\leq C(d) \sum_{\ell \geq 0} |S_O| (|S_O|+v s+\ell+1)^d e^{-\kappa \ell} \sup_{x \in \Lambda} \sum_{\substack{S \in \mathcal{P}_c(\Lambda) \\  x \in S}} \Vert A_S(\lambda \nu s) \Vert_{\mathrm{op}} \Vert O \Vert_{\mathrm{op}} e^{\kappa |S|} \\
	& \leq C(d) |S_O| (S_O+v s)^d \sum_{\ell\geq 0} \left( 1+\frac{\ell+1}{|S_O|+v s} \right)^d e^{-\kappa \ell} \sup_{x \in \Lambda} \sum_{\substack{S \in \mathcal{P}_c(\Lambda) \\  x \in S}} \Vert A_S(\lambda \nu s) \Vert_{\mathrm{op}} \Vert O \Vert_{\mathrm{op}} e^{\kappa |S|}\,.
	\end{split}
	\]
	Since the series in $\ell$ is convergent and its general term is bounded by $(2 + \ell)^d e^{-\kappa \ell}$, we define $C(d,\kappa):= C(d) \sum_{\ell \geq 0} (2+\ell)^d e^{-\kappa \ell}$ and we obtain
	\[
	\sum_{x \in \Lambda \setminus B_{S_O}} \sum_{\substack{S \in \mathcal{P}_c(\Lambda) \\  x \in S}} \Vert A_S(\lambda \nu s) \Vert_{\mathrm{op}} \Vert O \Vert_{op} |S_O| e^{\kappa |S|} e^{-\kappa d(x,B_{S_O})} \leq |S_O| (|S_O|+v  s)^d C(d,\kappa) \Vert O \Vert_{\mathrm{op}} \Vert A(\lambda \nu s) \Vert_{\kappa} \, .
	\]
	Then, for fixed $ s$, since $C(d,\kappa)\geq 1,$ combining the above estimate with \eqref{eq:salsa.verde} and \eqref{eq:mondeghili}, one has
	\[
	\Vert [A(\lambda \nu s), e^{-\ii  s Z} O e^{\ii  s Z} ] \Vert_{\mathrm{op}}  \leq 3|S_O| (|S_O|+v  s)^d C(d,\kappa) \Vert O \Vert_{\mathrm{op}} \Vert A(\lambda \nu s) \Vert_{\kappa}.
	\]
	Recalling that $\sup_{ s \in \mathbb{R}} \Vert A(\lambda \nu s)\Vert_\kappa \leq \Vert A \Vert_{\kappa,C^0}$, one has
	\[
	\int_0^t \ud  s \Vert[A(\lambda \nu s),e^{-\ii  s Z} O e^{\ii  s Z} ] \Vert_{\mathrm{op}} \leq 3|S_O|\Vert O \Vert_{\mathrm{op}} \Vert A \Vert_{\kappa,C^0} \frac{C(d,\kappa)}{(d+1) v } \left[(|S_O|+v t)^{d+1}-|S_O|^{d+1}  \right]\,,
	\]
        One now observes that
        \[
            \begin{split}
            (|S_O|+y)^{d+1}-|S_O|^{d+1} &= \sum_{j=0}^{d+1}\binom{d+1}{j} |S_O|^{d+1-j} y^j-|S_O|^{d+1} = \sum_{j=1}^{d+1} \binom{d+1}{j} |S_O|^{d+1-j}y^j \\
            &\leq C_1(d) |S_O|^{d+1} \langle y \rangle^{d} y\,.
            \end{split}
        \]
	Putting $y=v t$, and recalling that $v=C_{0}(d)\kappa^{-{(d+2)}}e^{\kappa}\Vert Z \Vert_{2 \kappa}$, one gets
	\[
		\begin{split}
			\int_0^t \ud  s \Vert[A,e^{-\ii  s Z} O e^{\ii  s Z} ] \Vert_{\mathrm{op}} & \leq 3 |S_O|^{d+2} \frac{C(d,\kappa) C_1(d)}{d+1} \langle v \rangle^d \langle t \rangle^d t \Vert O \Vert_{\operatorname{op}} \Vert A \Vert_{\kappa,C^0}\\
			&\leq C(d,\kappa,|S_O|) \langle \Vert Z \Vert_{2 \kappa} \rangle^d \langle t \rangle^{d+1} \Vert O \Vert_{\mathrm{op}} \Vert A \Vert_{\kappa,C^0} \, ,
		\end{split}
	\]
	with  $C(d,\kappa,|S_O|)=3 |S_O|^{d+2} \frac{C(d,\kappa) C_1(d)}{d+1} \langle C_{0}(d) \kappa^{-(d+2)} e^\kappa \rangle^d$.

\end{proof}

\begin{lemma}\label{lem.spagnolette}
	Let $H_{\obs}(\lambda \nu t) := \lambda H^{(\fin)}(\lambda \nu t)$, $\lambda$ as in Proposition \ref{normal.form}, and $H_{\mathrm{eff}}$ as in \eqref{eq:heff}. For any local observable $O$, there exists a constant $C_0 = C_0(|S_O|, d, \kappa,p, \| H_0 \|_{\kappa}, \|V \|_{\kappa, C^p})$ such that, for any $t \in \mathbb{R}$ we have
	\begin{equation}\label{bagigio}
	\Vert {U}^*_{H_{\obs}}(t) O {U}_{H_{\obs}}(t) - e^{\ii H_{\mathrm{eff}} t} O e^{-\ii H_{\mathrm{eff}} t} \Vert_{\mathrm{op}} \; \leq \; C_0 \|O\|_{\operatorname{op}} \lambda^{-\frac{p(1-b-\epsilon)}{\tau}} \langle t\rangle^{d+1}\,.
	\end{equation}
\end{lemma}
\begin{proof}
    The proof follows again closely the one of Lemma 6.4 of \cite{Gallone-Langella-2024}. First, one notes that, provided $\lambda$ is large enough, $ \frac{1}{2} \| H_0\|_{\kappa_\fin} \leq \|H_{\eff}\|_{\kappa_\fin} \leq 2   \| H_0 \|_{\kappa_\fin} $. We now define the auxiliary operator
    $$
    W(t',t) :=   U_{H_\obs}^*(t') e^{\ii (t-t') H_{\eff}} O e^{-\ii(t-t') H_{\eff}} U_{H_\obs}(t') - e^{\ii t H_{\eff}} O e^{-\ii t H_{\eff} t}\,, \quad t, t' \in \R\,,
    $$
    and we observe that $W(t,t) = U_{H_\obs}^*(t) O U_{H_\obs}(t) - e^{\ii t H_{\eff}} O e^{-\ii t H_{\eff} t}$, and $W(0, t) = 0$. Then, writing $W(t', t) = \int_{0}^{t'} \partial_{s} W(s, t) \ud s$, one gets
    \begin{equation}
    \begin{aligned}
          &U_{H_\obs}^*(t) O U_{H_\obs}(t) - e^{\ii t H_{\eff}} O e^{-\ii t H_{\eff} t} \\
    &\quad = \ii  \int_0^{t} U_{H_\obs}^*(s) \left[V_\lambda^{(\fin)}(\lambda \nu s) + R_\lambda^{(\fin)}(\lambda \nu s),\ e^{\ii (t-s) H_\eff} O e^{-\ii (t-s) H_\eff} \right] U_{H_\obs}(s) \ud s\\
    &\quad = \ii \int_0^{t} U_{H_\obs}^*(t-s) \left[V_\lambda^{(\fin)}(\lambda \nu (t-s)) + R_\lambda^{(\fin)}(\lambda \nu (t-s)),\ e^{\ii s H_\eff} O e^{-\ii s H_\eff} \right] U_{H_\obs}(t-s) \ud s\,.
    \end{aligned}
    \end{equation}
    We then apply Lemma \ref{lem:LemmaSempreDiverso} with $\kappa \leadsto \frac{\kappa_\fin}{2}$, $Z = H_\eff$ and $A = A_t$ defined by $A_t(\lambda \nu s) := V_\lambda^{(\fin)}(\lambda \nu (t-s)) + R_\lambda^{(\fin)}(\lambda \nu (t-s))$ for any $t$ and $s \in \R$. Note that for any $t \in \R$ $\|A_t \|_{\kappa, C^0} = \|V_\lambda^{(\fin)} + R_\lambda^{(\fin)}\|_{\kappa, C^0}$. Then we obtain
    \begin{align*}
        \Big\| U^*_{H_{\obs}}(t ) O &U_{H_{\obs}}(t) - e^{\ii H_\eff t} O e^{-\ii H_\eff t} \Big\|_{\operatorname{op}}\\
        &= \left\| \int_0^t U_{H_\obs}^*( t-s) \left[V_\lambda^{(\fin)}(\lambda \nu (t-s)) + R_\lambda^{(\fin)}(\lambda \nu(t- s))\,, e^{\ii H_\eff s} O e^{-\ii H_\eff s} \right] U_{H_\obs}(t-s) \ud  s \right\|_{\operatorname{op}}\\
        &\leq  \int_0^t \left\| \left[V_\lambda^{(\fin)}(\lambda \nu (t-s)) + R_\lambda^{(\fin)}(\lambda \nu (t-s))\,, e^{\ii H_\eff s} O e^{-\ii H_\eff s} \right] \right\|_{\operatorname{op}} \ud  s \\
        &\leq C(|S_O|,d,\textstyle{\frac{\kappa_{\fin}}{2}}) \langle \Vert H_\eff \Vert_{\kappa_\fin} \rangle^{d} \langle t \rangle^{d + 1} \|O\|_{\operatorname{op}} \|V_\lambda^{(\fin)} + R_\lambda^{(\fin)} \|_{\frac{\kappa_\fin}{2}, C^0}\\
        &\leq 3^d C(|S_O|,d,\textstyle{\frac{\kappa_{\fin}}{2}})  \langle \Vert H_0 \Vert_{\kappa_\fin} \rangle^d D  \langle t \rangle^{d + 1} \|O\|_{\operatorname{op}} \lambda^{-\frac{p(1-b-\epsilon)}{\tau}} \|V\|_{\kappa, C^p}\,,
    \end{align*}
    where $D$ is the positive constant appearing in Proposition \ref{normal.form} and where $\Vert H_\eff \Vert_{\kappa_\fin} \leq 2 \Vert H_0 \Vert_{\kappa_\fin}$. Estimate \eqref{bagigio} then follows with $C_0 =  3^d C(|S_O|,d,\textstyle{\frac{\kappa_{\fin}}{2}})  D   \langle \| H_0 \|_{\kappa_\fin} \rangle^d \|V \|_{ \kappa, C^p}$.
\end{proof}

\begin{lemma}\label{lemma:nemo}
	For any local observable $O$, there exists a constant $C_1 = C_1(|S_O|, d, \kappa, p, \| V \|_{\kappa, C^p})$ such that
    \begin{equation}\label{ragazza}
		\Vert U_H^*(t) O U_H(t) - {U}_{H_{\obs}}^*(t) O {U}_{H_{\obs}}(t) \Vert_{\mathrm{op}} \leq {C_1} \|O\|_{\operatorname{op}} \lambda^{-b} \, \qquad \forall t \in \mathbb{R}\,. 
		\end{equation}
\end{lemma}
\begin{proof}
    By Item 3 of Proposition \ref{normal.form}, one has (using the notation of the proof of Proposition \ref{normal.form}) and denoting $Y^{(-1)}:=\mathbbm{1}$,
    	\[
	\begin{split}
	\Vert U_H^*(t) O U_H(t) &- {U}_{H_{\obs}}^*(t) O {U}_{H_{\obs}}(t) \Vert_{\mathrm{op}} = \Vert U_H^*(t) (O-Y^*(\lambda \nu t)OY(\lambda \nu t)) U_H(t) \Vert_{\mathrm{op}} \\ 
	&= \Vert O - Y^*(\lambda \nu t) O Y(\lambda \nu t) \Vert_{\mathrm{op}} \\
	& \leq \sum_{n={0}}^{n_*-1} \Vert (Y^{(n)})^*(\lambda \nu t) O Y^{(n)}(\lambda \nu t) - (Y^{(n-1)})^*(\lambda \nu t) O Y^{(n-1)}(\lambda \nu t) \Vert_{\mathrm{op}} \\
	&=\sum_{n={0}}^{n_*-1} \Vert (Y^{(n-1)})^*(\lambda \nu t) (e^{-\ii G^{(n)}(\lambda \nu t)} O e^{\ii G^{(n)}(\lambda \nu t)}-O) Y^{(n-1)}(\lambda \nu t) \Vert_{\mathrm{op}} \\
	&=\sum_{n={0}}^{n_*-1} \Vert e^{-\ii G^{(n)}(\lambda \nu t)} O e^{\ii G^{(n)}(\lambda \nu t)} - O \Vert_{\mathrm{op}}\,.
	\end{split}
	\]
	We are thus left with finding a good bound for $\Vert e^{-\ii G^{(n)}(\lambda \nu t)} O e^{\ii G^{(n)}(\lambda \nu t)} - O \Vert_{\mathrm{op}}$. 
    For any fixed $t_0:=t$, we write
	\[
	\begin{split}
	\Vert e^{-\ii G(\lambda \nu t_0)} O e^{\ii G(\lambda \nu t_0)} - O \Vert_{\mathrm{op}} &\leq \int_0^1 \ud  s \Vert [ G(\lambda \nu t_0),e^{-\ii  s G(\lambda \nu t_0)} O e^{\ii  s G(\lambda \nu t_0)}] \Vert_{\mathrm{op}}\,
	\end{split}
	\]
	and we apply Lemma \ref{lem:LemmaSempreDiverso} with $Z=G(\lambda \nu t_0)$, $\kappa=\frac{\kappa_{\fin}}{2}$ and $t=1$. This yields
	\[
	\begin{split}
	\Vert e^{-\ii G(\lambda \nu t_0)} O e^{\ii G(\lambda \nu t_0)} - O \Vert_{\mathrm{op}} &\;\leq\;  2^{d+1} C(|S_O|,d,{\textstyle \frac{\kappa_{\fin}}{2}}) \langle \Vert G(\lambda \nu t) \Vert_{ \kappa_{\fin}} \rangle^d \|O\|_{\mathrm{op}} \|G(\lambda \nu t_0)\|_{\frac{\kappa_{\fin}}{2}} \, ,
	\end{split}
	\]
    which gives
    	\begin{equation}
	\Vert e^{-\ii G(\lambda \nu t)} O e^{\ii G(\lambda \nu t)} - O \Vert_{\mathrm{op}} \leq 2^{d+1} C(|S_O|,d, {\textstyle \frac{\kappa_{\fin}}{2}}) \langle \Vert G(\lambda \nu t) \Vert_{ \kappa_{\fin}} \rangle^d \Vert O \Vert_{\mathrm{op}} \Vert G(\lambda \nu t) \Vert_{\frac{\kappa_{\fin}}{2}}\,.
	\end{equation}
    	Now, using \eqref{che.bella.g}  and \eqref{fatta.cosi} recalling $\kappa_{\fin} \leq \kappa_n \leq \kappa$, one has 
        	\begin{equation}\label{san.martino}
	\| G^{(n)}(\lambda \nu t)\|_{\kappa_{\fin}} \leq \Vert G^{(n)} \Vert_{\kappa_n, \sigma}\leq \hat{C} \frac{ K^\tau}{\gamma \lambda} \lambda^{-b} e^{-n} \Vert V \Vert_{\kappa,C^p}\,,
	\end{equation}
	(where $\hat{C}$ is the constant appearing in \eqref{fatta.cosi}) whence, since $K$ is taken as in \eqref{cento.passi}, if $\lambda$ is large enough one has $\hat{C} \frac{K^\tau}{\gamma \lambda} \leq 1$ and
	\[
	\begin{split}
	\sum_{n=0}^{n_*-1} \Vert e^{-\ii G^{(n)}(\lambda \nu t)} O e^{\ii G^{(n)}(\lambda \nu t)}-O \Vert_{\mathrm{op}} &\leq 2^{d+1} C(|S_O|,d,{\textstyle \frac{\kappa_{\fin}}{2}}) \Vert O \Vert_{\mathrm{op}} \Vert V \Vert_{\kappa, C^p} \lambda^{-b} \sum_{n=0}^{n_*-1} e^{-n} \\
	&\leq  2^{d+1} C(|S_O|,d,{\textstyle \frac{\kappa_{\fin}}{2}}) \Vert O \Vert_{\mathrm{op}} \frac{{e}}{e-1} \Vert V \Vert_{\kappa, C^p} \lambda^{-b} \, .
	\end{split}
	\]
	This proves \eqref{ragazza}.  
\end{proof}
\begin{proof}[Proof of Theorem \ref{thm:main}-(ii)]
Combining estimates \eqref{bagigio} and \eqref{ragazza}, one gets
\begin{align*}
    \| U_H^*(t) O U_H(t) - e^{\ii H_\eff t} O e^{-\ii H_\eff t} \|_{\operatorname{op}} &\leq \|U_H^*(t) O U_H(t) - U_{H_{\obs}}^*(t) O U_{H_{\obs}}(t)\|_{\operatorname{op}}\\
    & \quad + \|  U_{H_{\obs}}^*(t) O U_{H_{\obs}}(t) - e^{\ii t H_\eff t} O e^{-\ii t H_\eff}\|_{\operatorname{op}}\\
    & \leq C_2 \|O\|_{\operatorname{op}} (\lambda^{-\frac{p(1-b)-\epsilon}{\tau}} \langle t \rangle^{d+1} + \lambda^{-b})\\
    &\leq 2  C_2 \|O\|_{\operatorname{op}} \lambda^{-b}\,
\end{align*}
with $C_2 := \max\{C_0,\ C_1\}$, $C_0$ and $C_1$ are the constants appearing respectively in Lemma \ref{lem:LemmaSempreDiverso} and Lemma \ref{lem.spagnolette}, and using the assumption that $|t| \leq \lambda^{\frac{-b + \frac{p}{\tau}(1-b)-\epsilon}{d+1}}$.
\end{proof}
\section{Breaking of the prethermal regime}
The goal of this section is to prove the following result which is the core of the proof of Theorem \ref{thm:fine.di.mondo}.
\begin{proposition}\label{prop:ricrescita}
   Let $n = 2$ and $\tau = 1 + \epsilon,$ with $\epsilon >0$. For any $\gamma >0$, there exist a Diophantine vector $\nu \in DC^2(\gamma, \tau)$, and sequences $\{\lambda_m\}_{m \in \N} \subset \R^+$, and $\{k_m\}_{m \in \N} \in \Z^2 \setminus \{0\}$, $k_m \equiv (k_m^{(1)},\ k_m^{(2)})$ for any $m$, such that the sequence of Hamiltonians
   \begin{equation}\label{eq:esempi}
       \mathtt{H}_{m}(\lambda_m \nu t) := \sigma^{(3)} + \mathtt{V}_{m}(\lambda_m \nu t)\,, \quad \mathtt{V}_{m}(\varphi) := \frac{2}{|k_m|^{p}} \cos(k^{(1)}_m \varphi_1) \cos(k^{(2)}_m \varphi_2)  \sigma^{(1)} \quad \forall \varphi \in \T^2\,,
   \end{equation}
   where $\sigma^{(3)} := \begin{pmatrix}
       1 & 0 \\ 0 & -1
   \end{pmatrix}$ and $\sigma^{(1)} := \begin{pmatrix}
       0 & 1 \\ 1 & 0
   \end{pmatrix}$ are Pauli matrices, acting on $\mathfrak{h}=\mathbb{C}^2$, satisfies the following:
   \begin{itemize}
    \item[(i)]  $\lambda_m \to \infty$ as $m \to \infty$;
       \item[(ii)] $\|\mathtt{V}_m\|_{C^p(\T^n; \cB(\frak h))} \in \left [ 2^{1-p} \,,\ 2 \right]$ for any $m \in \N$;
       \item[(iii)] Defining for any $m$ the magnetization $M_m(t)$ at time $t \in \R$ 
 as        \begin{equation}\label{eq:magma}
           M_m(t) :=  \langle \sigma^{(3)} U_{\tH_m}(t) \psi_0, U_{\tH_m}(t) \psi_0 \rangle\,, \quad \psi_0 := \begin{pmatrix} 0 \\ 1 \end{pmatrix}\,,
       \end{equation}
       there exist $\{t_m\}_m \subset \R^+$ such that 
       \begin{equation}
           M_m(t_m) - M_m(0) \geq  \frac1 2 \quad \text{at} \quad t_m \in [C_1 \lambda_m^{\frac{p}{\tau}},\  C_2 \lambda_m^{\frac{p}{\tau} + \epsilon}]\,,
       \end{equation}
       with $C_1 := \frac{\pi}{4}\left(\frac{\gamma}{2}\right)^{\frac{p}{\tau}}$ and $C_2 := \frac{\pi}{4} |\nu|^{\frac{2 p}{\tau}} $. 
   \end{itemize}
\end{proposition}
Before proving Proposition \ref{prop:ricrescita}, we start with showing how it implies Theorem \ref{thm:fine.di.mondo}.
\begin{proof}[Proof of Theorem \ref{thm:fine.di.mondo}]
Note that, for any $\Psi \in \cH_\Lambda$ with $\|\Psi\|_{\cH_\Lambda} = 1$, for any $t \in \R$ and for any $H$ self-adjoint, one has
\begin{equation}\label{eq:op.la}
    \begin{split}
    \| U_H^*(t)\langle H \rangle U_H(t) - \langle H \rangle\|_{\operatorname{op}} &\geq \langle \left(U^*_H(t)\langle H \rangle U_H(t) - \langle H \rangle \right)\Psi,  \Psi \rangle \\
    &= \langle \langle H \rangle U_{H}(t) \Psi, U_H(t) \Psi \rangle - \langle \langle H \rangle \Psi, \Psi \rangle\,,
    \end{split}
\end{equation}
where $\langle H \rangle$ denotes the average of $H$ over the angles \eqref{eq:mediamente}.
Therefore it is sufficient to choose the vector $\nu \in \R^n$ and the sequences $\{\lambda_m\}_{m}$, $\{ k_m\}_{m}$ and $\{t_m\}_{m}$ as in Proposition \ref{prop:ricrescita}, the Hamiltonians
\begin{equation}\label{eq:big.mac}
    \begin{aligned}
        H_m(\lambda_m \nu t) &:= \sum_{x \in \Lambda} \tH_{m, x} (\lambda_m \nu t)\,, \quad
        \tH_{m, x} (\lambda_m \nu t):= \sigma_x^{(3)} + \tV_{m,x}(\lambda \nu_m t)\,,\\
        \tV_{m,x}(\lambda \nu_m t) &:= \frac{2}{|k_m|^{p}} \cos(\lambda_m \nu_1 k^{(1)}_m t) \cos(\lambda_m \nu_2 k^{(2)}_m t) \sigma^{(1)}_x\,, 
    \end{aligned}
\end{equation}
and to define the state
$$
\Psi := \bigotimes_{x \in \Lambda} \psi_x\,, \quad \psi_x := \begin{pmatrix}
    0 \\ 1
\end{pmatrix} \quad \forall x \in \Lambda\,.
$$
Then one has $H_m = H_0 + V_m(\lambda_m \nu t)\,,$ with
$$
H_0 = \langle H_m \rangle = \sum_{x \in \Lambda} \sigma^{(3)}_x\,, \quad V_m = \sum_{x \in \Lambda} \tV_{m,x}\,.
$$
Moreover, for any $\kappa >0$
$$
\| H_0\|_{\kappa} = \sup_{x \in \Lambda} \|\sigma^{(3)}_x\|_{\operatorname{op}} e^{\kappa} = e^{\kappa}\,, 
$$
and by Item (ii) of Proposition \ref{prop:ricrescita}
$$
\| V_m\|_{\kappa, C^p} = \sup_{x \in \Lambda} \|\tV_{m, x}\|_{C^p(\T^n; \cB(\mathfrak{h}))} e^{\kappa} \in [2^{1-p}e^{\kappa}\,,\ 2 e^{\kappa}]\,.
$$
Due to the fact that $H_m$ so constructed is the sum of non-interacting single particle Hamiltonians $\tH_{m,x}$, one also has
$$
 U_{H_m} (t) \Psi = \bigotimes_{x \in \Lambda} U_{\tH_{m,x}}(t) \psi_x\,.
$$
Thus, using Proposition \ref{prop:ricrescita}, one has
\begin{align*}
&\langle H_0 U_{H_m}(t_m) \Psi, U_{H_m}(t_m) \Psi \rangle - \langle H_0 \Psi, \Psi \rangle \\
& =   \sum_{x \in \Lambda} \left(\langle \sigma^{(3)}_x U_{\tH_{m,x}}(t_m) \psi_x, U_{\tH_{m,x}}(t_m)\psi_x \rangle  - \langle \sigma^{(3)}_x \psi_x, \psi_x \rangle\right) \geq \sum_{x \in \Lambda} \frac{1}{2} = \frac{|\Lambda|}{2}\,,
\end{align*}
which gives \eqref{eq:crescendo}.
\end{proof}
\begin{remark}
    The same result stated in Proposition \ref{prop:ricrescita} holds in a slightly simpler way and within the same time scales if we replace the perturbations $\tV_m$ in \eqref{eq:esempi} with
    $$
    \mathtt{W}_m(\lambda_m \nu t):= \frac{1}{|k_m|^p} \cos(k_m \cdot \nu t) \sigma^{(1)}\,.
    $$
    However, we chose to present the construction with $\tV_m$ as in \eqref{eq:esempi}, since the latter ones are genuinely time quasi-periodic functions of $t$, whereas $\mathtt{W}_m$ are actually time periodic. Note that there is no contradiction with the result in \cite{Abanin2017}, which claims that in the time periodic case stability holds for exponentially long times in $|\lambda|$ independently of the regularity in time, since the perturbations $\mathtt{W}_m$ have diverging periods as $m \to \infty$.
\end{remark}
The remaining part of the present section is devoted to the proof of Proposition \ref{prop:ricrescita}.  We start to prove Item (ii), which is immediate: for any choice of $\{k_m\}_{m} \subset \Z^2 \setminus \{0\}$, one has
\[
    \begin{split}
    \| \tV_m\|_{C^p(\T^n; \cB(\frak h))} &= \sup_{0 \leq |p'| \leq p} \sup_{\varphi \in \T^2} \|\partial^{p'}_\varphi \tV_m(\varphi)\| = \sup_{0 \leq |p'| \leq p} 2 |k_m|^{-p} (\max\{|k^{(1)}_m|,\ |k^{(2)}_m|\})^{|p'|}  \\ &= 2 |k_m|^{-p} (\max\{|k^{(1)}_m|,\ |k^{(2)}_m|\})^{p}\,.
    \end{split}
\]
Then Item (ii) follows observing that $\max\{|k^{(1)}_m|,\ |k^{(2)}_m|\} \in [\frac1 2 |k_m|, |k_m|]$\,. The remaining points are more delicate, and we shall prove in the next subsections.
\subsection{Almost resonances}
We start with exhibiting our choice of the sequences $\{\lambda_m\}_m$ and $\{ k_m\}_m$. In order to do this, we fix a constant $\tau := 1 + \epsilon$, $\epsilon>0$, and  a Diophantine vector $\nu \in DC^2(\gamma, \tau)$ of the form \begin{equation}\label{eq:Nudista}
    \nu = (\alpha, 1)\,, \quad \alpha \in \R^+\,.
\end{equation}
We then recall the following well-known result:
\begin{theorem}[Dirichlet approximation theorem] \label{thm:dirichlet} For any $\alpha \in \R$ there exist an increasing sequence $\{q_m\}_m \subset \N$ and a sequence $\{p_m\}_m \subset \Z$ such that
\begin{equation}
    \left| \alpha - \frac{p_m}{q_m}\right| \leq \frac{1}{q_m^2} \quad \forall m \in \N\,.
\end{equation}
\end{theorem}
From Dirichlet Approximation Theorem \ref{thm:dirichlet} we can immediately deduce the following:
\begin{lemma}[Best approximants sequence] \label{lem:bestini}
    Let $\tau := 1 + \epsilon$ with $\epsilon>0$ and let $\alpha \in \R$ be a positive number such that $\nu := (\alpha, 1) \in DC^2(\gamma, \tau)$ for some $\gamma >0$\,. Then there exists a sequence $\{k_m\}_{m} \subset \Z^2$  such that $|k_m| \to \infty$ as $m \to \infty$ and one has \begin{equation}\label{eq:nu.dot.k}
        \frac{\gamma}{|k_m|^{\tau}} \leq | \nu \cdot k_m | \leq \frac{2 |\nu|}{|k_m|^{\tau - \epsilon}} \quad \forall m \in \N\,.
    \end{equation}
\end{lemma}
\begin{proof}
    It is sufficient to choose $k_m = (q_m, - p_m)$ for any $m \in \N$, with $\{p_m\}_m$ and $\{q_m\}_m$ as in Theorem \ref{thm:dirichlet}. Then the first inequality in \eqref{eq:nu.dot.k} holds due to the fact that the vector $\nu$ is Diophantine, whereas to prove the second inequality we observe that, by Theorem \ref{thm:dirichlet}, for any $m$ one has
    \begin{equation} \label{eq:small.div}
    |\nu \cdot k_m| = |\alpha q_m - p_m| \leq \frac{1}{|q_m|}\,. 
    \end{equation}
    Moreover, since $\left|\alpha - \frac{p_m}{q_m}\right| \leq \frac{1}{q_m^2} \leq 1$, one also has
    \[
        |p_m| \leq |\alpha q_m-p_m|+ \alpha q_m \leq 1+\alpha q_m \leq (1+\alpha)q_m
    \]
    which implies
    \begin{equation}
    \label{eq:equiv}
    |q_m| \leq |k_m| := |p_m| + q_m \leq (2 +  \alpha) q_m \leq 2 |\nu| \, q_m \,.
    \end{equation}
    Therefore, combining \eqref{eq:small.div} and \eqref{eq:equiv}, one deduces
    $$
    |\nu \cdot k_m| \leq \frac{1}{q_m} \leq \frac{2 |\nu|
    }{|k_m|} = \frac{ 2 |\nu|}{|k_m|^{\tau - \epsilon}}\,.
    $$
\end{proof} 
\begin{remark}
    As pointed out in Remark \ref{rmk:sono.tanti}, the set of vectors $\nu = (\alpha, 1)$ satisfying the assumptions of Lemma \ref{lem:bestini} has full Lebesgue measure, thus in particular it is non empty. Note indeed that $(\nu_1, \nu_2) \in DC^2(\gamma, \tau)$ if and only if $(\tfrac{\nu_1}{\nu_2}, 1) \in DC^2(\tfrac{\gamma}{\nu_2}, \tau)$, and that we are requiring $\tau >1$.
\end{remark}
Therefore, we fix
\begin{equation}
    k_m \text{ as in Lemma \ref{lem:bestini}}\,, \quad \lambda_m := \frac{2}{|\nu \cdot k_m|} \quad \forall m \in \N\,,
\end{equation}
so that, by estimate \eqref{eq:nu.dot.k}, one has
\begin{gather}
      \label{eq:il.lambda.che.ci.piace}
      \frac{1}{|\nu|} |k_m|^{\tau - \epsilon} \leq \lambda_m \leq \frac{2}{\gamma} |k_m|^{\tau}\,.
\end{gather}
We also deduce the following result. 

\begin{lemma}\label{lem:tutto.grande}
    Let $k=(k^{(1)},k^{(2)}) \in \mathbb{Z}^2$, $\nu \in \R^2$ and $\lambda \in \R^+$. If $| k | \geq \left(\frac{\gamma}{2} \right)^{\frac{1}{\tau}} \lambda^{\frac{1}{\tau}}$, $|\nu \cdot k|=\frac{2}{\lambda}$  and ${\lambda > \left(\frac{8}{\min\{\nu_1,\nu_2\}}\right)^\tau \frac{2}{\gamma}}$, then
    \begin{equation} \label{eq:kappiniDaSoli}
        |k^{(1)}| \geq C \lambda ^{\frac{1}{\tau}} \quad \text{and} \quad |k^{(2)}| \geq C \lambda^{\frac{1}{\tau}} \, ,
    \end{equation}
    with $C=\min\left\{\frac{\nu_1}{4 \nu_2},\frac{\nu_2}{4 \nu_1}\right\} \left( \frac{\gamma}{2} \right)^{\frac{1}{\tau}}$.
\end{lemma}
\begin{proof}
    Since $| k |  \geq \left(\frac{\gamma}{2}\right)^{\frac{1}{\tau}} \lambda^{\frac{1}{\tau}}$, then at least one among $|k^{(1)}|$ and $|k^{(2)}|$ has to be greater or equal $\frac{1}{2}\left(\frac{\gamma}{2}\right)^{\frac{1}{\tau}} \lambda^{\frac{1}{\tau}}$ otherwise, by triangular inequality, one obtains a contradiction. 

    Let us assume, without loss of generality, that $|k^{(1)}| \geq \frac{1}{2}\left(\frac{\gamma}{2}\right)^{\frac{1}{\tau}} \lambda^{\frac{1}{\tau}}.$ Then,
    \begin{equation}
        \begin{split}
            |k^{(2)}| & = \frac{|\nu_2 k^{(2)}|}{\nu_2} = \frac{|\nu_2 k^{(2)} + \nu_1 k^{(1)} - \nu_1 k^{(1)}|}{\nu_2} \\
            &= \frac{|\nu_1 k^{(1)}-\nu\cdot k|}{\nu_2} \geq \frac{|\nu_1 k^{(1)}| - |\nu \cdot k|}{\nu_2} \geq \frac{1}{2}\frac{\nu_1}{\nu_2} \left(\frac{\gamma}{2}\right)^{\frac{1}{\tau}} \lambda^{\frac{1}{\tau}}-\frac{2}{\nu_2} \, .
        \end{split}
    \end{equation}
    Using now that $\lambda > \left(\frac{8}{\nu_1}\right)^\tau \frac{2}{\gamma}$, we obtain the thesis.
\end{proof}
\begin{lemma}\label{lem:OmeghiniDalBasso}
    Let $\nu = (\alpha, 1)$ as in Lemma \ref{lem:bestini}, $\omega^\pm_m := \alpha \lambda_m k^{(1)}_m  - \lambda_m k^{(2)}_m  \pm 2$, with $k_m = (k^{(1)}_m, k^{(2)}_m)$ as in Lemma \ref{lem:bestini}. Then, for any $m \in \N$ such that $\lambda_m>(\frac{4}{C|\nu|})^{\tau}$, we have
    \begin{equation}
        |\omega^{\pm}_m| \geq \frac{C |\nu|}{2} \lambda_m^{\frac{1}{\tau}}\,,
    \end{equation}
    with $C>0$  the same constant in \eqref{eq:kappiniDaSoli}.
\end{lemma}
\begin{proof}
Noting that since the vector $\nu$ has positive components, the requirements $k^{(1)}_m,\ k^{(2)}_m \geq C \lambda_m^{\frac{1}{\tau}}$ (which follows from Lemma \ref{lem:tutto.grande}) and $|\nu \cdot k_m| =\frac{2}{\lambda_m}$ imply that $k^{(1)}_m$ and $k^{(2)}_m$ must have opposite sign. Therefore, 
    \begin{equation}
        |\omega^\pm_m| \geq |\alpha \lambda_m k^{(1)}_m - \lambda_m k^{(2)}_m| - 2 = |\alpha k^{(1)}_m |+| k^{(2)}_m|-2 \geq \frac{C(\alpha +1)}{2}\lambda_m^{\frac{1}{\tau}} = \frac{C |\nu|}{2} \lambda_m^{\frac{1}{\tau}}
    \end{equation}
    where in the last step, we used that $\lambda_m \geq \left(\frac{4}{C(\alpha+1)}\right)^{\tau} = \left(\frac{4}{C|\nu|}\right)^{\tau}$.
\end{proof}

\subsection{Time evolution}
In order to explicitly compute the magnetization $M_m(t)$ for any $m$, we 
need to compute the time evolution of the state $\psi_0$ defined in \eqref{eq:magma}, namely
\begin{equation}
    \psi_m(t) := U_{\tH_m}(t) \psi_0\,.
\end{equation}
This is obtained passing for any $m \in \N$ to variable
\begin{equation}\label{eq:im.better}
    \phi_m(t) := e^{\ii \sigma^{(3)} t} \psi_m(t)\,,
\end{equation}
which, recalling the definition of $\tH_m$ in \eqref{eq:esempi} and our choices of $\tV_m$ in \eqref{eq:esempi} and $\nu$, solves $\forall t$ the equation
\begin{equation}\label{eq:gauge.me.away}
\begin{aligned}
     \ii \partial_t \phi_m(t) &= \frac{2}{|k_m|^{p}} \cos(\lambda_m \nu_1 k^{(1)}_m t) \cos(\lambda_m \nu_2 k^{(2)}_m t)  e^{\ii \sigma^{(3)} t}\sigma^{(1)}  e^{-\ii \sigma^{(3)} t}\phi_m(t)\\
     & = \frac{2}{|k_m|^{p}} \left(\cos(2 t) + \cos\big(\lambda_m (\alpha k_m^{(1)}- k_m^{(2)} t)\big) \right) e^{\ii \sigma^{(3)} t}\sigma^{(1)}  e^{-\ii \sigma^{(3)} t}\phi_m(t)\,,
     \\ \phi_m(0) &= \psi_0\,.
\end{aligned}
\end{equation}
The following result is useful in order to compute the time evolution $\phi_m(t)$.

\begin{lemma}\label{lem:ConiugiPesciosi}
    Let $j,k,l$ be three different indices of Pauli matrices. Then,
    \begin{equation}
        e^{\ii \sigma^{(j)} t} \sigma^{(k)} e^{-\ii \sigma^{(j)}t}=\cos(2t) \sigma^{(k)}-\epsilon_{jkl} \sin(2t) \sigma^{(l)} \, ,
    \end{equation}
    where $\epsilon_{jkl}$ is the Levi-Civita symbol.
\end{lemma}
\begin{proof}
    Recalling the notation in \eqref{eq:commuta.che.ti.passa}, by  conjugation one has
    \begin{equation}
        \begin{split}
            e^{\ii \sigma^{(j)}t} \sigma^{(k)} e^{- \ii \sigma^{(j)}t} & = \sum_{r=0}^{+\infty} \frac{(\ii t)^r}{r!} \mathrm{Ad}_{\sigma^{(j)}}^r \sigma^{(k)} \\
            &=\sum_{\substack{r=0\\ r \, \text{even}}}^{+\infty} \frac{( \ii t)^r}{r!} \mathrm{Ad}_{\sigma^{(j)}}^r \sigma^{(k)} + \sum_{\substack{r=0\\ r \, \text{odd}}}^{+\infty} \frac{( \ii t)^r}{r!} \mathrm{Ad}_{\sigma^{(j)}}^r \sigma^{(k)} \\
            &=\sum_{r=0}^{+\infty} (-1)^{r}\frac{ t^{2r}}{(2r)!} \mathrm{Ad}_{\sigma^{(j)}}^{2r} \sigma^{(k)} + \ii \sum_{r=0}^{+\infty} (-1)^r \frac{t^{2r+1}}{(2r+1)!} \mathrm{Ad}_{\sigma^{(j)}}^{2r}\left([\sigma^{(j)},\sigma^{(k)}] \right)=(\star)
        \end{split}
    \end{equation}
    Using now that for any $j,k,l \in \{1,2,3\}$
    \begin{equation}
        [\sigma^{(j)},\sigma^{(k)}]= 2 \ii \epsilon_{jkl} \sigma^{(l)} \, , \qquad \mathrm{Ad}_{\sigma^{(j)}}^{2r} \sigma^{(k)}=4^r \sigma^{(k)}\,,
    \end{equation}
    one has
    \begin{equation}
            (\star)=\sum_{r=0}^{+\infty} (-1)^r \frac{(2t)^{2r}}{(2r)!} \sigma^{(k)}-\epsilon_{jkl} \sum_{r=0}^{+\infty} (-1)^r \frac{(2t)^{2r+1}}{(2r+1)!} \sigma^{(l)} 
    \end{equation}
    which is the thesis upon recognizing the Taylor series of sine and cosine functions.
\end{proof}

We use Lemma \ref{lem:ConiugiPesciosi} to deduce the following:
\begin{lemma}\label{lem:savana}
    For any $m \in \N$ let $\phi_m$ as in \eqref{eq:gauge.me.away}, then $\forall t \in \R$, for $\lambda>\left(\frac{300}{C| \nu| }\right)^\tau$ we have
    \begin{gather}
        \phi_m(t) = e^{-\ii \frac{t}{2 |k_{m}|^{p}} \sigma^{(1)}} \psi_0 + r_m(t)\,, \label{eq:seidicannove}\\
        \| r_m(t)\|_{\frak{h}} \leq \frac{1}{ |k_{m}|^{p}} + \frac{2|t|}{ |k_{m}|^{2p}}\,. \label{eq:seiventi}
    \end{gather}
\end{lemma}
\begin{proof}
    Using Lemma \ref{lem:ConiugiPesciosi} to compute the Hamiltonian in \eqref{eq:gauge.me.away}, and simplifying with Werner's formulas one has
    \begin{equation}
        \begin{split}
            \ii \partial_t \phi_m(t) &= \frac{1}{2 |k_m|^{p}} \Big(\sigma^{(1)}+\cos(4t) \sigma^{(1)}+\cos(\omega^+_m t) \sigma^{(1)}+\cos(\omega^-_m t) \sigma^{(1)} \\ & \qquad\qquad\qquad -\sin(4t) \sigma^{(2)}-\sin(\omega^+_m t) \sigma^{(2)}-\sin(\omega^-_m t) \sigma^{(2)}
\Big)\phi_m(t) ,
        \end{split}
    \end{equation}
    where $\omega^\pm_m := \alpha \lambda_m k^{(1)}_m - \lambda_m k^{(2)}_m \pm 2$.
    We then define $\xi_m(t):= e^{\ii \sigma^{(1)} \frac{t}{2 |k_{m}|^{p}}} \phi_m(t)$, which evolves according to
    \begin{equation}\label{eq:csi.miami}
        \ii \partial_t \xi_m(t) = Q_m(t) \xi_m(t) \, ,
    \end{equation}
    where 
    \begin{equation}\label{eq:tante.care.cose}
        \begin{split}
            Q_m(t)&:=\frac{1}{2|k_m|^p} \Big[\left(\cos(4t) +\cos(\omega^+_m t) +\cos(\omega^-_m t)\right) \sigma^{(1)}\\ & \qquad - \left(\sin(4t) \cos({\textstyle \frac{t}{|k_m|^p}}) + \sin(\omega^+_mt)\cos({\textstyle \frac{t}{|k_m|^p}}) + \sin(\omega^-_m t)\cos({\textstyle \frac{t}{|k_m|^p}}) \right)\sigma^{(2)} \\
            &\qquad + \left(\sin(4t) \sin({\textstyle \frac{t}{|k_m|^p}}) +\sin(\omega^+_m t) \sin({\textstyle \frac{t}{|k_m|^p}}) + \sin(\omega^-_mt) \sin({\textstyle \frac{t}{|k_m|^p}}) \right)\sigma^{(3)}\Big] \, .
        \end{split}
    \end{equation}
    We note that
    \begin{equation}\label{eq:StimaQGrande}
        \sup_{t \in \mathbb{R}} \Vert Q_m(t) \Vert_{\mathrm{op}} \leq \frac{9}{2 |k_m|^{p}} \, .
    \end{equation}
    The explicit solution of \eqref{eq:csi.miami} is
    \begin{equation}
        \xi_m(t)=\xi_m(0)-\ii  \int_0^s Q_m(s) \xi_m(s) \, ds = \xi_m(0)+ q_m(t)\,,
    \end{equation}
    and we now are going to show that
\begin{equation}\label{eq:normina}
        \Vert q_m(t) \Vert_{\mathfrak{h}} \leq \frac{5}{8 |k_m|^{p}}+\frac{ 10 \lambda_{m}^{-\frac{1}{\tau}}}{C |\nu| | k_m |^p }+ \frac{31 |t|}{ 16 |k_m|^{2p}}+\frac{31 \lambda_{m}^{-\frac{1}{\tau}}|t|}{ C |\nu| | k_m |^{2p} } \, .
    \end{equation}
    We only estimate the terms in the first line of \eqref{eq:tante.care.cose}, the others having analogous upper bounds.  For the first addendum, integrating by parts, one has
    \begin{equation}
         \frac{1}{2 |k_m|^{p}}\int_0^t \cos(4s) \sigma^{(1)} \xi_m(s) \, ds = \frac{1}{8 |k_m|^{p}} \sin(4t) \sigma^{(1)}\xi_m(t) +\frac{\ii}{8 |k_m|^{p}} \int_0^t \sin(4s) \sigma^{(1)} Q_m(s) \xi_m(s) \, ds \, .
    \end{equation}
    Computing the $\frak{h}$ norm, one gets the uniform bounds in time
    \begin{equation}\label{eq:molti}
        \left\Vert \frac{1}{8 |k_m|^{p}} \sin(4t) \sigma^{(1)}\xi_m(s)  \right\Vert_{\mathfrak{h}} \leq \frac{1}{8 |k_m|^{p}}
    \end{equation}
    and
    \begin{equation}\label{eq:conti}
        \left\Vert \frac{\ii}{8 |k_m|^{p}} \int_0^t \sin(4s) \sigma^{(1)} Q_m(s) \xi_m(s) \, ds \right\Vert_{\mathfrak{h}} \leq \frac{1}{8 |k_m|^{p}} \sup_{s \in \mathbb{R}} \Vert Q_m(s) \Vert_{\mathrm{op}} |t| \leq \frac{9|t|}{16 |k_m|^{2p}}\,.
    \end{equation}
    For the second and third addendum, integrating by parts, one has
    \begin{equation}
        \begin{split}
            \frac{1}{2 |k_m|^p} &\int_0^t
            \cos(\omega^\pm_m t) \sigma^{(1)} \xi_m(s) \, ds \\
            &=\frac{1}{2 \omega^\pm_m |k_m|^p} \left( \sin(\omega^\pm_m t) \sigma^{(1)} \xi_m(t) +\ii \int_0^t \sin(\omega^\pm_m s) \sigma^{(1)} Q_m(s) \xi_m(s) \right)\, ds\,.
        \end{split}
    \end{equation}
    We use now Lemma \ref{lem:OmeghiniDalBasso} to estimate $|\omega^\pm_m| \geq \frac{C|\nu|}{2} \lambda_m^{\frac{1}{\tau}}$ and arguing as above we get the estimate
    \begin{equation}\label{eq:molto}
        \left\Vert \frac{1}{2 |k_m|^p} \int_0^t
            \cos(\omega^\pm_m t) \sigma^{(1)} \xi_m(s) \, ds \right\Vert_{\mathfrak{h}} \leq \frac{ \lambda_m^{-\frac{1}{\tau}}}{C|\nu| |k_m|^p}+\frac{9|t| \lambda_m^{-\frac{1}{\tau}}}{2 C |\nu| |k_m|^{2p}}\,.
    \end{equation}
    Using analogous reasonings, one has
    \begin{equation}\label{eq:contosi}
        \begin{split}  
            \left\Vert \frac{1}{2 |k_m|^p} \int_0^t \sin(4s) \cos(\textstyle{\frac{s}{|k_m|^p}}) \sigma^{(2)} \xi_m(s) \, \ud s \right\Vert_{\mathfrak{h}} &\leq \frac{1}{4 |k_m|^p}+\frac{11 |t|}{16 |k_m|^{2p}}\,, \\
            \left\Vert \frac{1}{2 |k_m|^p  }\int_0^t \sin(\omega^\pm_m s) \cos(\textstyle{\frac{s}{|k_m|^p}}) \sigma^{(2)} \xi_m(s) \, ds \right\Vert_{\mathfrak{h}} &\leq \frac{2 \lambda_m^{-\frac{1}{\tau}}}{C|\nu| |k_m|^p}+\frac{11 |t|\lambda_m^{-\frac{1}{\tau}}}{2 C|\nu| |k_m|^{2p}}\,, \\
            \left\Vert \frac{1}{2 |k_m|^p} \int_0^t \sin(4s) \sin(\textstyle{\frac{s}{|k_m|^p}}) \sigma^{(3)} \xi_m(s) \, \ud s \right\Vert_{\mathfrak{h}} &\leq \frac{1}{4 |k_m|^p}+\frac{11 |t|}{16 |k_m|^{2p}}\,, \\
            \left\Vert \frac{1}{2 |k_m|^p  }\int_0^t \sin(\omega^\pm_m s) \sin(\textstyle{\frac{s}{|k_m|^p}}) \sigma^{(3)} \xi_m(s) \, ds \right\Vert_{\mathfrak{h}} &\leq \frac{2 \lambda^{-\frac{1}{\tau}}}{C|\nu| |k_m|^p}+\frac{11 |t|\lambda^{-\frac{1}{\tau}}}{2 C|\nu| |k_m|^{2p}} \,.
        \end{split}        
    \end{equation}
    Using now that $\lambda>\left(\frac{300}{C|\nu|} \right)^\tau$ and combining \eqref{eq:molti}, \eqref{eq:conti}, \eqref{eq:molto}, and \eqref{eq:contosi}, we get
    \begin{equation}\label{eq:SecondaNormina}
        \Vert q_m(t) \Vert_{\mathfrak{h}} \leq \frac{1}{|k_m|^p}+\frac{2|t|}{|k_m|^{2p}} \, .
    \end{equation}  
    
     We are now ready to prove the statement. Indeed, unfolding back the gauge transformations and noting that $\xi_m(0)=\phi_m(0)=\psi_m(0)$, one has
     \begin{equation}
     	\phi_m(t)=e^{-\ii \sigma^{(1)} \frac{t}{2 |k_m|^{p}}} \phi_m(0)+e^{-\ii \sigma^{(1)} \frac{t}{2 |k_m|^{p}}} q_m(t)
     \end{equation}
     Calling now $r_m(t)=e^{-\ii \sigma^{(1)} \frac{t}{2 |k_m|^{p}}} q_m(t)$, we proved \eqref{eq:seidicannove} and since $e^{-\ii \sigma^{(1)} \frac{t}{2 |k_m|^{p}}}$ is unitary, $\Vert r_m(t) \Vert_{\mathfrak{h}} = \Vert q_m(t) \Vert_{\mathfrak{h}}$ that can be bounded by \eqref{eq:SecondaNormina} and this proves \eqref{eq:seiventi}.
\end{proof}
Once we know $\phi_m(t)$, we are ready to prove Proposition \ref{prop:ricrescita}:
\begin{proof}[Proof of Proposition \ref{prop:ricrescita}]
Recalling definitions of $M_m$ and $\phi_m$ in \eqref{eq:magma} and \eqref{eq:im.better}, one has
$$
M_m(t) = \langle \sigma^{(3)} \psi_m(t), \psi_m(t) \rangle = \langle \sigma^{(3)} e^{-\ii \sigma^{(3)}t } \phi_m(t), e^{- \ii \sigma^{(3)}t } \phi_m(t) \rangle = \langle \sigma^{(3)} \phi_m(t), \phi_m(t) \rangle\,.
$$
Using Lemma \ref{lem:savana}, we compute
\begin{equation}\label{eq:m'n'm's}
    \begin{aligned}
    M_m(t) & = \langle \sigma^{(3)} e^{-\ii \frac{t}{2 |k_m|^{p}} \sigma^{(1)}} \psi_0, e^{-\ii \frac{t}{2 |k_m|^{p}} \sigma^{(1)}} \psi_0 \rangle + \rho_m(t)\,,\\
    \rho_m(t)& := 2 \Ree\,\langle \sigma^{(3)} e^{-\ii \frac{t}{2 |k_m|^{p}} \sigma^{(1)}} \psi_0, r_m(t) \rangle + \langle \sigma^{(3)} r_m(t), r_m(t) \rangle\,.
\end{aligned}
\end{equation}
Now, up to taking $|t| \leq |k_m|^{p}$ and $|k_m|$ large enough, from \eqref{eq:seiventi} one has $\|r_m(t)\|_{\mathfrak{h}} \leq \frac{3}{|k_m|^{p}} \leq 1,$ therefore
\begin{equation}\label{eq:e.piccolo}
    |\rho_m(t)| \leq 2 \|r_m(t)\|_{\mathfrak{h}} + \|r_m(t)\|_{\mathfrak{h}}^2 \leq 3 \|r_m(t)\|_{\mathfrak{h}} \leq \frac{9}{ |k_m|^{p}}\,.
\end{equation}
Moreover, using Lemma \ref{lem:ConiugiPesciosi}, one has
\begin{equation}\label{eq:oscilla}
   \begin{aligned}
    \langle \sigma^{(3)} e^{-\ii \frac{t}{|k_m|^{p}} \sigma^{(1)}} \psi_0, e^{-\ii \frac{t}{|k_m|^{p}} \sigma^{(1)}} \psi_0 \rangle & = \left \langle \cos\left(\frac{2 t}{|k_m|^{p}} \right) \sigma^{(3)} \psi_0, \psi_0 \right\rangle + \left\langle \sin \left( \frac{2t}{|k_m|^{p}}\right) \sigma^{(2)} \psi_0, \psi_0\right\rangle\\
    & =  -\cos\left(\frac{2 t}{|k_m|^{p}} \right) \,.
\end{aligned} 
\end{equation}
Then, choosing $t_m = \frac{\pi}{4} |k_m|^{p}$ and combining \eqref{eq:m'n'm's},  \eqref{eq:oscilla}, and \eqref{eq:e.piccolo}, we have $\cos({2 t_m}|k_m|^{-p} ) = 0$ and
$$
M_m(t) - M_m(0) = 1-\cos\left( \frac{2 t_m}{ |k_m|^{p}} \right) + \rho_m(t_m) \geq 1 - |\rho_m(t_m)| \geq \frac{1}{2}\,,
$$
up to increasing again the size of $|k_m|$ .
Then to prove Propostion \ref{prop:ricrescita} it remains to observe that, by estimate \eqref{eq:il.lambda.che.ci.piace}, one has
$$
t_m \in \left[ \frac{\pi}{4}\left(\frac{\gamma}{2}\right)^{\frac{p}{\tau}}\lambda_m^{\frac{p}{\tau}},\ \frac{\pi}{4} |\nu|^{\frac{p}{\tau - \epsilon}} \lambda_m^{\frac{p}{\tau - \epsilon}}\right] \subset \left[ \frac{\pi}{4}\left(\frac{\gamma}{2}\right)^{\frac{p}{\tau}}\lambda_m^{\frac{p}{\tau}},\ \frac{\pi}{4} |\nu|^{\frac{2 p}{\tau}}  \lambda_m^{\frac{p}{\tau} + \epsilon}\right]\,.
$$

\vspace{-0.4cm}
\end{proof}

\vspace{0.4cm}
\noindent
\textbf{Acknowledgments.} We acknowledge support of the Italian Group of Mathematical Physics GNFM (INdAM). B.L.~is supported by PRIN 2022 (2022HSSYPN). We thank Ricardo~Grande for pointing out stimulating comments during the final phase of the preparation of this paper.

\footnotesize


\begin{thebibliography}{10}

\bibitem{Abanin2017}
{\sc D.~Abanin, W.~De~Roeck, W.~W. Ho and F.~Huveneers}, {\em {A Rigorous
  Theory of Many-Body Prethermalization for Periodically Driven and Closed
  Quantum Systems}}, Communications in Mathematical Physics, 354 (2017),
  p.~809–827.

\bibitem{Bambusi2020}
{\sc D.~Bambusi, B.~Grébert, A.~Maspero and D.~Robert}, {\em {Growth of
  Sobolev norms for abstract linear Schr\"{o}dinger equations}}, Journal of the
  European Mathematical Society, 23 (2020), p.~557–583.

\bibitem{Bambusi-Langella-Cinf}
{\sc D.~Bambusi and B.~Langella}, {\em {A $C^\infty$ Nekhoroshev theorem}},
  Mathematics in Engineering, 3 (2021), p.~1–17.

\bibitem{barbieri-farre}
{\sc S.~Barbieri and G.~Farré}, {\em {Nearly-Optimal Effective Stability
  Estimates Around Diophantine Tori of H\"{o}lder Hamiltonians}}, Journal of
  Dynamics and Differential Equations, 37 (2024), p.~3185–3196.

\bibitem{Barbieri-Marco-Massetti}
{\sc S.~Barbieri, J.-P. Marco and J.~E. Massetti}, {\em {Analytic Smoothing and
  Nekhoroshev Estimates for H\"{o}lder Steep Hamiltonians}}, Communications in
  Mathematical Physics, 396 (2022), p.~349–381.

\bibitem{Beatrez2021}
{\sc W.~Beatrez et~al.}, {\em {Floquet Prethermalization with Lifetime
  Exceeding 90 s in a Bulk Hyperpolarized Solid}}, Physical Review Letters, 127
  (2021).

\bibitem{Benettin-Galgani-Giorgilli-1987}
{\sc G.~Benettin, L.~Galgani and A.~Giorgilli}, {\em {Realization of holonomic
  constraints and freezing of high frequency degrees of freedom in the light of
  classical perturbation theory. Part I}}, Communications in Mathematical
  Physics, 113 (1987), p.~87–103.

\bibitem{Benettin-Galgani-Giorgilli-1989}
{\sc G.~Benettin, L.~Galgani and A.~Giorgilli}, {\em {Realization of holonomic
  constraints and freezing of high frequency degrees of freedom in the light of
  classical perturbation theory. Part II}}, Communications in Mathematical
  Physics, 121 (1989), p.~557–601.

\bibitem{Benettin-Gallavotti}
{\sc G.~Benettin and G.~Gallavotti}, {\em {Stability of motions near resonances
  in quasi-integrable Hamiltonian systems}}, Journal of Statistical Physics, 44
  (1986), p.~293–338.

\bibitem{Bianchini2025}
{\sc R.~Bianchini, L.~Franzoi, R.~Montalto and S.~Terracina}, {\em {Large
  Amplitude Quasi-Periodic Traveling Waves in Two Dimensional Forced Rotating
  Fluids}}, Communications in Mathematical Physics, 406 (2025).

\bibitem{Blekher1992}
{\sc P.~M. Blekher, H.~R. Jauslin and J.~L. Lebowitz}, {\em {Floquet spectrum
  for two-level systems in quasiperiodic time-dependent fields}}, Journal of
  Statistical Physics, 68 (1992), p.~271–310.

\bibitem{Bounemoura2010}
{\sc A.~Bounemoura}, {\em {Nekhoroshev estimates for finitely differentiable
  quasi-convex Hamiltonians}}, Journal of Differential Equations, 249 (2010),
  p.~2905–2920.

\bibitem{Bounemoura2011}
{\sc A.~Bounemoura}, {\em {Effective Stability for Gevrey and Finitely
  Differentiable Prevalent Hamiltonians}}, Communications in Mathematical
  Physics, 307 (2011), p.~157–183.

\bibitem{bou.AHP}
{\sc A.~Bounemoura}, {\em {Optimal stability and instability for near-linear
  Hamiltonians}}, Ann. Henri Poincar{\'e}, 13 (2012), pp.~857--868.

\bibitem{Bukov2015}
{\sc M.~Bukov, L.~D’Alessio and A.~Polkovnikov}, {\em {Universal
  high-frequency behavior of periodically driven systems: from dynamical
  stabilization to Floquet engineering}}, Advances in Physics, 64 (2015),
  p.~139–226.

\bibitem{Chierchia-LN}
{\sc L.~Chierchia}, {\em {KAM Lectures}}, in Dynamical Systems part I, N.~E.
  del libro, ed., Pubbl. Cent. Ric. Mat. Ennio Giorgi, 2003, pp.~pp. 1--55.

\bibitem{Choi2020}
{\sc J.~Choi et~al.}, {\em {Robust Dynamic Hamiltonian Engineering of Many-Body
  Spin Systems}}, Physical Review X, 10 (2020).

\bibitem{Choi2017}
{\sc S.~Choi et~al.}, {\em {Observation of discrete time-crystalline order in a
  disordered dipolar many-body system}}, Nature, 543 (2017), p.~221–225.

\bibitem{Ciampa2024}
{\sc G.~Ciampa, R.~Montalto and S.~Terracina}, {\em {Large amplitude traveling
  waves for the non-resistive MHD system}}, Journal of Hyperbolic Differential
  Equations, 21 (2024), p.~707–790.

\bibitem{Pros}
{\sc W.~De~Roeck, F.~Huveneers and O.~A. Prośniak}, {\em {Long Persistence of
  Localization in a Disordered Anharmonic Chain Beyond the Atomic Limit}},
  Communications in Mathematical Physics, 406 (2025).

\bibitem{DeRoeck-Verreet}
{\sc W.~De~Roeck and V.~Verreet}, {\em {Very slow heating for weakly driven
  quantum many-body systems}}, 2019.

\bibitem{Dodson}
{\sc M.~M. Dodson, B.~P. Rynne and J.~A.~G. Vickers}, {\em {Averaging in
  multifrequency systems}}, Nonlinearity, 2 (1989), p.~137–148.

\bibitem{Dumitrescu2018}
{\sc P.~T. Dumitrescu, R.~Vasseur and A.~C. Potter}, {\em {Logarithmically Slow
  Relaxation in Quasiperiodically Driven Random Spin Chains}}, Physical Review
  Letters, 120 (2018).

\bibitem{Dutta2025}
{\sc P.~Dutta, S.~Choudhury and V.~Shukla}, {\em {Prethermalization in the PXP
  model under continuous quasiperiodic driving}}, Physical Review B, 111
  (2025).

\bibitem{DAlessio2014}
{\sc L.~D’Alessio and M.~Rigol}, {\em {Long-time Behavior of Isolated
  Periodically Driven Interacting Lattice Systems}}, Physical Review X, 4
  (2014).

\bibitem{Eckardt2005}
{\sc A.~Eckardt, C.~Weiss and M.~Holthaus}, {\em {Superfluid-Insulator
  Transition in a Periodically Driven Optical Lattice}}, Physical Review
  Letters, 95 (2005).

\bibitem{Eisert2015}
{\sc J.~Eisert, M.~Friesdorf and C.~Gogolin}, {\em {Quantum many-body systems
  out of equilibrium}}, Nature Physics, 11 (2015), p.~124–130.

\bibitem{Else2016}
{\sc D.~V. Else, B.~Bauer and C.~Nayak}, {\em {Floquet Time Crystals}},
  Physical Review Letters, 117 (2016).

\bibitem{Else2020}
{\sc D.~V. Else, W.~W. Ho and P.~T. Dumitrescu}, {\em {Long-Lived Interacting
  Phases of Matter Protected by Multiple Time-Translation Symmetries in
  Quasiperiodically Driven Systems}}, Physical Review X, 10 (2020).

\bibitem{Erds2024}
{\sc L.~Erdős, J.~Henheik, J.~Reker and V.~Riabov}, {\em {Prethermalization
  for Deformed Wigner Matrices}}, Annales Henri Poincaré, 26 (2024),
  p.~1991–2033.

\bibitem{Fang2025}
{\sc J.~Fang, Q.~Zhou and X.~Wen}, {\em {Phase transitions in quasiperiodically
  driven quantum critical systems: Analytical results}}, Physical Review B, 111
  (2025).

\bibitem{farre}
{\sc G.~Farré}, {\em {On the optimal effective stability bounds for
  quasi-periodic tori of finitely differentiable and Gevrey Hamiltonians}},
  Archiv der Mathematik, 122 (2023), p.~213–225.

\bibitem{farre-fayad}
{\sc G.~Farré and B.~Fayad}, {\em {Instabilities of invariant quasi-periodic
  tori}}, Journal of the European Mathematical Society, 24 (2022),
  p.~4363–4383.

\bibitem{Franzoi}
{\sc L.~Franzoi}, {\em {Reducibility for a linear wave equation with Sobolev
  smooth fast-driven potential}}, Discrete and Continuous Dynamical Systems, 43
  (2023), p.~3251–3285.

\bibitem{Franzoi-Maspero}
{\sc L.~Franzoi and A.~Maspero}, {\em {Reducibility for a fast-driven linear
  Klein–Gordon equation}}, Annali di Matematica Pura ed Applicata (1923 -),
  198 (2019), p.~1407–1439.

\bibitem{Gallone-Langella-2024}
{\sc M.~Gallone and B.~Langella}, {\em {Prethermalization and Conservation Laws
  in Quasi-Periodically Driven Quantum Systems}}, Journal of Statistical
  Physics, 191 (2024).

\bibitem{Geier2021}
{\sc S.~Geier et~al.}, {\em {Floquet Hamiltonian engineering of an isolated
  many-body spin system}}, Science, 374 (2021), p.~1149–1152.

\bibitem{Harper2020}
{\sc F.~Harper, R.~Roy, M.~S. Rudner and S.~Sondhi}, {\em {Topology and Broken
  Symmetry in Floquet Systems}}, Annual Review of Condensed Matter Physics, 11
  (2020), p.~345–368.

\bibitem{He2023}
{\sc G.~He et~al.}, {\em {Quasi-Floquet Prethermalization in a Disordered
  Dipolar Spin Ensemble in Diamond}}, Physical Review Letters, 131 (2023).

\bibitem{He2025}
{\sc G.~He et~al.}, {\em {Experimental Realization of Discrete Time
  Quasicrystals}}, Physical Review X, 15 (2025).

\bibitem{Ho2023}
{\sc W.~W. Ho, T.~Mori, D.~A. Abanin and E.~G. Dalla~Torre}, {\em {Quantum and
  classical Floquet prethermalization}}, Annals of Physics, 454 (2023),
  p.~169297.

\bibitem{Ho2018}
{\sc W.~W. Ho, I.~Protopopov and D.~A. Abanin}, {\em {Bounds on Energy
  Absorption and Prethermalization in Quantum Systems with Long-Range
  Interactions}}, Physical Review Letters, 120 (2018).

\bibitem{Holthaus2015}
{\sc M.~Holthaus}, {\em {Floquet engineering with quasienergy bands of
  periodically driven optical lattices}}, Journal of Physics B: Atomic,
  Molecular and Optical Physics, 49 (2015), p.~013001.

\bibitem{Kumar2024}
{\sc S.~Kumar and S.~Choudhury}, {\em {Prethermalization in aperiodically
  driven classical spin systems}}, Physical Review E, 110 (2024).

\bibitem{Kundu2025}
{\sc A.~Kundu, T.~Nag and A.~Rajak}, {\em {Statistical prethermalization in a
  randomly kicked many-body classical rotor system}}, Physical Review B, 112
  (2025).

\bibitem{Kyprianidis2021}
{\sc A.~Kyprianidis et~al.}, {\em {Observation of a prethermal discrete time
  crystal}}, Science, 372 (2021), p.~1192–1196.

\bibitem{Zhuo-preprint}
{\sc Z.-Y. Li and Y.-R. Zhang}, {\em {Prethermal discrete time crystals in
  one-dimensional classical Floquet systems with nearest-neighbor
  interactions}}, 2025.

\bibitem{QP-esperimento2}
{\sc Z.-H. Liu et~al.}, {\em {Prethermalization by Random Multipolar Driving on
  a 78-Qubit Superconducting Processor}}, 2025.

\bibitem{Long2021}
{\sc D.~M. Long, P.~J. Crowley and A.~Chandran}, {\em {Nonadiabatic Topological
  Energy Pumps with Quasiperiodic Driving}}, Physical Review Letters, 126
  (2021).

\bibitem{Machado2020}
{\sc F.~Machado et~al.}, {\em {Long-Range Prethermal Phases of Nonequilibrium
  Matter}}, Physical Review X, 10 (2020).

\bibitem{Machado2019}
{\sc F.~Machado et~al.}, {\em {Exponentially slow heating in short and
  long-range interacting Floquet systems}}, Physical Review Research, 1 (2019).

\bibitem{Marco-Sauzin}
{\sc J.-P. Marco and D.~Sauzin}, {\em {Stability and instability for Gevrey
  quasi-convex near-integrable hamiltonian systems}}, Publications
  Mathématiques de l’IHÉS, 96 (2003), p.~199–275.

\bibitem{AltroQP}
{\sc D.~Marripour and J.~Abouie}, {\em {From time crystals to time
  quasicrystals: Exploring novel phases in transverse field Ising chains}},
  2025.

\bibitem{Xiao-Mi-2021}
{\sc X.~Mi et~al.}, {\em {Time-crystalline eigenstate order on a quantum
  processor}}, Nature, 601 (2021), p.~531–536.

\bibitem{Moessner2017}
{\sc R.~Moessner and S.~L. Sondhi}, {\em {Equilibration and order in quantum
  Floquet matter}}, Nature Physics, 13 (2017), p.~424–428.

\bibitem{Mori2021}
{\sc T.~Mori et~al.}, {\em {Rigorous Bounds on the Heating Rate in Thue-Morse
  Quasiperiodically and Randomly Driven Quantum Many-Body Systems}}, Physical
  Review Letters, 127 (2021).

\bibitem{Moser1960}
{\sc J.~Moser}, {\em {On the elimination of the irrationality condition and
  Birkhoff's concept of complete stability}}, Boletín de la Sociedad
  Matemática Mexicana, 5 (1960), pp.~167--174.

\bibitem{Neistadt}
{\sc A.~I. Neishtadt}, {\em {Averaging in Multifrequency Systems II}}, Doklady
  Akademii Nauk SSSR, 226 (1976), pp.~1295--1298.

\bibitem{Nekhoroshev1977}
{\sc N.~N. Nekhoroshev}, {\em {An Exponential Estimate of the Time of Stability
  of Nearly-Integrable Hamiltonian Systems}}, Russian Mathematical Surveys, 32
  (1977), p.~1–65.

\bibitem{Nek79}
{\sc N.~N. Nekhoroshev}, {\em {An exponential estimate of the time of stability
  of nearly integrable Hamiltonian systems. II}}, Trudy Seminara imeni I. G.
  Petrovskogo, 5 (1979), pp.~5--50.

\bibitem{Oka2019}
{\sc T.~Oka and S.~Kitamura}, {\em {Floquet Engineering of Quantum Materials}},
  Annual Review of Condensed Matter Physics, 10 (2019), p.~387–408.

\bibitem{Peng2021}
{\sc P.~Peng et~al.}, {\em {Floquet prethermalization in dipolar spin chains}},
  Nature Physics, 17 (2021), p.~444–447.

\bibitem{Potirniche2017}
{\sc I.-D. Potirniche et~al.}, {\em {Floquet Symmetry-Protected Topological
  Phases in Cold-Atom Systems}}, Physical Review Letters, 119 (2017).

\bibitem{Qi2021}
{\sc Z.~Qi, G.~Refael and Y.~Peng}, {\em {Universal nonadiabatic energy pumping
  in a quasiperiodically driven extended system}}, Physical Review B, 104
  (2021).

\bibitem{Randall2021}
{\sc J.~Randall et~al.}, {\em {Many-body–localized discrete time crystal with
  a programmable spin-based quantum simulator}}, Science, 374 (2021),
  p.~1474–1478.

\bibitem{Rechtsman2013}
{\sc M.~C. Rechtsman et~al.}, {\em {Photonic Floquet topological insulators}},
  Nature, 496 (2013), p.~196–200.

\bibitem{RubioAbadal2020}
{\sc A.~Rubio-Abadal et~al.}, {\em {Floquet Prethermalization in a Bose-Hubbard
  System}}, Physical Review X, 10 (2020).

\bibitem{Rudner2020}
{\sc M.~S. Rudner and N.~H. Lindner}, {\em {Band structure engineering and
  non-equilibrium dynamics in Floquet topological insulators}}, Nature Reviews
  Physics, 2 (2020), p.~229–244.

\bibitem{Salamon2004}
{\sc D.~A. Salamon}, {\em {The Kolmogorov-Arnold-Moser theorem.}}, Mathematical
  Physics Electronic Journal [electronic only], 10 (2004), pp.~Paper No. 3, 37
  p.--Paper No. 3, 37 p.

\bibitem{Singh2019}
{\sc K.~Singh et~al.}, {\em {Quantifying and Controlling Prethermal
  Nonergodicity in Interacting Floquet Matter}}, Physical Review X, 9 (2019).

\bibitem{Weidinger2017}
{\sc S.~A. Weidinger and M.~Knap}, {\em {Floquet prethermalization and regimes
  of heating in a periodically driven, interacting quantum system}}, Scientific
  Reports, 7 (2017).

\bibitem{Yao2017}
{\sc N.~Y. Yao, A.~C. Potter, I.-D. Potirniche and A.~Vishwanath}, {\em
  {Discrete Time Crystals: Rigidity, Criticality, and Realizations}}, Physical
  Review Letters, 118 (2017).

\bibitem{Ye2021}
{\sc B.~Ye, F.~Machado and N.~Y. Yao}, {\em {Floquet Phases of Matter via
  Classical Prethermalization}}, Physical Review Letters, 127 (2021).

\bibitem{Zhang2017}
{\sc J.~Zhang et~al.}, {\em {Observation of a discrete time crystal}}, Nature,
  543 (2017), p.~217–220.

\bibitem{Zhao2021}
{\sc H.~Zhao, F.~Mintert, R.~Moessner and J.~Knolle}, {\em {Random Multipolar
  Driving: Tunably Slow Heating through Spectral Engineering}}, Physical Review
  Letters, 126 (2021).

\bibitem{Qp-experiment-2025}
{\sc D.-Y. Zhu et~al.}, {\em {Observation of Discrete Time Quasicrystal in
  Rydberg Atomic Gases}}, 2025.

\end{thebibliography}

\end{document}